%% file: new_arXiv.tex
\documentclass[authoryear]{imsart}

\usepackage[ruled]{algorithm2e}
\usepackage{amsfonts}
\usepackage{amsmath}
\usepackage{array}
\usepackage[english]{babel}
\usepackage{bbm}
\usepackage{graphicx}
\usepackage[utf8]{inputenc}
\usepackage{longtable}
\usepackage{natbib}
\usepackage{subcaption}
\usepackage{amsmath,amsfonts,amssymb,hyperref,graphicx,siunitx,fancyhdr,lastpage,enumerate, commath,amsthm} %, color}
\usepackage{url}

\input{macros.tex}

\startlocaldefs
\theoremstyle{plain}

\newtheorem{theorem}{Theorem}[section]
\newtheorem{lemma}[theorem]{Lemma}
\newtheorem{proposition}[theorem]{Proposition}
\theoremstyle{remark}

\newtheorem{assumption}{Assumption}
\newtheorem{remark}{Remark}

\endlocaldefs
\newcommand\correction[1]{#1}

\begin{document}

\begin{frontmatter}
\title{Sampling using adaptive regenerative processes}
\runtitle{Adaptive Restore}

\begin{aug}
\author[A]{\fnms{Hector}~\snm{McKimm}\ead[label=e1]{h.mckimm@imperial.ac.uk}}%\orcid{0000-0001-9922-6605}}
\author[B]{\fnms{\ Andi Q.}~\snm{Wang}\ead[label=e2]{andi.wang@warwick.ac.uk}}
\author[C]{\fnms{Murray}~\snm{Pollock}\ead[label=e3]{murray.pollock@newcastle.ac.uk}}
\author[B,D]{\fnms{\\ Christian P.}~\snm{Robert}\ead[label=e4]{\\ C.A.M.Robert@warwick.ac.uk}}%\orcid{0000-0001-6635-3261}}
\author[B]{\fnms{\ Gareth O.}~\snm{Roberts}\ead[label=e5]{Gareth.O.Roberts@warwick.ac.uk}}
%%%%%%%%%%%%%%%%%%%%%%%%%%%%%%%%%%%%%%%%%%%%%%
%% Addresses                                %%
%%%%%%%%%%%%%%%%%%%%%%%%%%%%%%%%%%%%%%%%%%%%%%
\address[A]{Department of Mathematics,
Imperial College London, London, United Kingdom\printead[presep={,\ }]{e1}}

\address[B]{Department of Statistics,
University of Warwick, Coventry, United Kingdom\printead[presep={,\ }]{e2,e4,e5}}

\address[C]{School of Mathematics,
Newcastle University, Newcastle-upon-Tyne, United Kingdom\printead[presep={,\ }]{e3}}

\address[D]{CEREMADE, Universit\'e Paris Dauphine PSL, Paris, France}
\end{aug}

\begin{abstract}
Enriching Brownian motion with regenerations from a fixed regeneration distribution $\mu$ at a particular regeneration rate $\kappa$ results in a Markov process that has a target distribution $\pi$ as its invariant distribution. For the purpose of Monte Carlo inference, implementing such a scheme requires firstly selection of regeneration distribution $\mu$, and secondly computation of a specific constant $C$. Both of these tasks can be very difficult in practice for good performance. We introduce a method for adapting the regeneration distribution, by adding point masses to it. This allows the process to be simulated with as few regenerations as possible and obviates the need to find said constant $C$. Moreover, the choice of fixed $\mu$ is replaced with the choice of the initial regeneration distribution, which is considerably less difficult. We establish convergence of this resulting self-reinforcing process and explore its effectiveness at sampling from a number of target distributions. The examples show that adapting the regeneration distribution guards against poor choices of fixed regeneration distribution and can reduce the error of Monte Carlo estimates of expectations of interest, especially when $\pi$ is skewed.
\end{abstract}


\end{frontmatter}
%

%%%%%%%%%%%%%%
% Introduction
%%%%%%%%%%%%%%

\section{Introduction}

Bayesian statistical problems customarily require computing quantities of the form $\pi[f] \equiv \mathbb{E}_\pi[f(X)]$, where $X$ is some random variable with distribution $\pi$ and $f$ is a function, meaning this expectation is the integral $\int_{\Rd} f(x) \pi(x) \dif x$, when $\Rd$ is the state space. For sophisticated models, it may be impossible to compute this integral analytically. Furthermore, it may be impractical to generate independent samples for use in Monte Carlo integration. In this case, Markov chain Monte Carlo (MCMC) methods \citep{robert2004monte} may be used to generate a Markov chain $X_0, X_1, \dots$ with limiting distribution $\pi$, and then approximate $\pi[f]$ by $n^{-1} \sum_{i=1}^n f(X_i)$.

The chain is constructed by repeatedly applying a collection of Markov transition kernels $P_1, \dots, P_m$, each satisfying $\pi P_i = \pi$ for $i = 1, \dots, m$. The Metropolis--Hastings algorithm \citep{metropolis1953equation, hastings1970monte} is normally used to construct and simulate from reversible $\pi$-invariant Markov transition kernels. A single kernel $P$ may be used to represent $P_1, P_2, \dots, P_m$, with form depending on whether a cycle is used ($P = P_1 P_2 \cdots P_m$) or a mixture ($P = [P_1 + P_2 + \cdots + P_m]/m$). Using multiple kernels allows different dynamics to be used, for example, by making transitions on the local and global scales.

The MCMC framework described above is restrictive. Firstly, each kernel must be $\pi$-invariant; for example, it is not possible for $P_1$ and $P_2$ to be individually non-$\pi$-invariant and yet somehow compensate for each other so that their combination $P$ is $\pi$-invariant. To achieve this $\pi$-invariance, each kernel is designed to be reversible. This acts as a further restriction; by definition, reversible kernels satisfy detailed balance and thus have diffusive dynamics. That is, chains generated using reversible kernels show random-walk-like behaviour, which is inefficient. Recently, there has been increasing interest in the use of non-reversible Markov processes for MCMC \citep{bierkens2019zig, bouchard-cote2018bouncy, Pollock2020}.

A further restriction of the typical MCMC framework is that it is difficult to make use of \emph{regeneration}. At \emph{regeneration times}, a Markov chain effectively starts again; its future is independent of its past. Regeneration is useful from both theoretical and practical perspectives. Nummelin's splitting technique \citep{nummelin1978splitting} may be used in MCMC algorithms to simulate regeneration events \citep{mykland1995regeneration, gilks1998adaptive}. However, the technique scales poorly: regenerations become exponentially rarer with dimension; see the discussion in \citet{gilks1998adaptive}.

An interesting direction to address these issues appeared in \cite{Wang2021}. The authors introduced the \emph{Restore process}, defined by enriching an underlying Markov process, which may not be $\pi$-invariant, with regenerations from some fixed \emph{regeneration distribution} $\mu$ at a \emph{regeneration rate} $\kappa$ so that the resulting Markov process is $\pi$-invariant. The state-dependent regeneration rate is given by $\kappa = \tilde{\kappa} + C \mu / \pi$, where $\tilde{\kappa}$ is a functional of the derivatives of $\log \pi$ and $C$ is a constant chosen so that $\kappa\ge 0$ pointwise. The segments of the process between regeneration times, known as \emph{tours}, are independent and identically distributed. When applied to Monte Carlo, we refer to this as the \emph{Restore sampler}. The process provides a general framework for using non-reversible dynamics, local and global moves, as well as regeneration within an MCMC sampler. Sample paths of the continuous-time process are used to form a Monte Carlo sum to approximate $\pi[f]$.

An issue with the Restore sampler is the choice $C \mu$. For a given $\mu$ and $\pi$, it is difficult to compute how large $C$ needs to be to guarantee $\kappa \ge 0$. Furthermore, it is not obvious how to set $\mu$ in the first place. This issue is in fact crucial, since the Restore sampler is very sensitive to the choice of $\mu$. One might consider adapting $\mu$ to become closer to $\pi$. Comparable discrete-time MCMC methods that use both local and global moves tend to resort to cycle kernels, for instance combining MALA \citep{roberts1996exponential} and Independence Sampler kernels. Typically, the proposal distribution for the Independence Sampler is initialized as a simple distribution, then adapted based on samples from $\pi$ to become closer to $\pi$. For example, \cite{gabrie2021efficient} define the proposal distribution using a normalizing flow \citep{papamakarios2021normalizing}, which is tuned to become closer to $\pi$. In simulating a Restore process, the strategy of adapting $\mu$ to become closer to $\pi$ could only work if the initial choice of $\mu$ were a `good' one that allows $\kappa$ to be small enough for sampling to be possible. In Section \ref{sec:adaptive_restore} we give an example of $\kappa$ being so large that, under a seemingly sensible choice of $\mu$, sampling is not possible due to numerical reasons. Furthermore in Section \ref{sec:examples}, we demonstrate that even if sampling is numerically possible, using a fixed regeneration distribution may make it inefficient.

In this work, we consider an adaptive approach so that a far smaller regeneration rate may be used. In particular, instead of using a fixed regeneration distribution $\mu$, we will use at time $t$ a regeneration distribution $\mu_t$, which is adapted so that it converges to the \textit{minimal regeneration distribution} $\mu^+$. This distribution typically has a \textit{compact} support, and will result in a considerably more robust algorithm, thus having a similar spirit to the proposed Barker proposal of \cite{Livingstone2022}. We call the novel Markov process an \emph{adaptive Restore process} and the original Restore process the \emph{standard Restore process}. An asymptotic analysis for large dimensions is presented in Supplement~D \citep{mckimm2024sampling} to justify this strategy. The regeneration distribution is initially a fixed parametric distribution $\mu_0$, then as the process is generated point masses are added, so that $\mu_t$ is a mixture of a parametric distribution and point masses. Throughout simulation, the \emph{minimal} regeneration rate is used. This rate is defined implicitly by a particular constant $C^+$, which need not be known explicitly. Thus adaptive Restore obviates the need to find a suitably large constant.

Adaptive Restore differs from adaptive MCMC methods \citep{andrieu2008tutorial, roberts2009examples, haario2001adaptive}, since the latter adapt the Markov transition kernel used in generating the Markov chain, whilst the former adapts the regeneration distribution.

Besides the methodological contributions of this work, from a theoretical perspective, this paper presents a novel application of the stochastic approximation technique to establishing convergence of self-reinforcing processes, as previously utilised in, say, \cite{Aldous1988, Benaim2018, Mailler2020}. In particular, we will adapt the proof technique of \cite{Benaim2018}---which applies to discrete-time Markov chains on a compact state space---to deduce validity of our adaptive Restore process, which is a continuous-time Markov process on a noncompact state space. This will be achieved by identifying a natural embedded discrete-time Markov chain, taking values on a compact subset, whose convergence implies convergence of the overall process. Theoretical summary and comparison are provided in Section~\ref{subsec:summary_theory}.

A secondary contribution of this article is showing that it is possible to use a standard Restore process to estimate the normalizing constant of an unnormalized density.

The rest of the article is arranged as follows. Section~\ref{sec:bmr} reviews standard Restore. Next, Section~\ref{sec:adaptive_restore} introduces the adaptive Restore process and its use as a sampler. Section~\ref{sec:theory} is a self-contained Section on the theory of adaptive Restore, where we prove its validity; see Section~\ref{subsec:summary_theory} for a summary of our theoretical contributions. Examples are then provided in Section~\ref{sec:examples} and Section~\ref{sec:discussion} concludes.

%%%%%%%%%%%%%%%%%%%%%
% The Restore process
%%%%%%%%%%%%%%%%%%%%%

\section{The Restore process} \label{sec:bmr}

This section describes the standard Restore process, as introduced in \cite{Wang2021}. We define the process, explain how it may be used to estimate normalizing constants, introduce the concept of minimal regeneration, and present the case where the underlying process is Brownian motion.

%%%%%%%%%%%%%%%%%%%%%%%%%%%%%%%%%%%%%%%%
% Regeneration-Enriched Markov Processes
%%%%%%%%%%%%%%%%%%%%%%%%%%%%%%%%%%%%%%%%

\subsection{Regeneration-enriched Markov processes}

The Restore process is defined as follows. Let $\{ Y_t \}_{t \ge 0}$ be a diffusion or jump process on $\R^d$. The regeneration rate $\kappa:\R^d \to [0,\infty)$, which we will define shortly, is locally bounded and measurable. Define the \emph{tour length} as
\begin{equation}
    \tourLength = \inf \left\{ t \ge 0 : \int_0^t \kappa(Y_s) \dif s \ge \xi \right\},
    \label{eq:tour_length}
\end{equation}
for $\xi \sim \text{Exp}(1)$ independent of $\{ Y_t \}_{t \ge 0}$.
%and by convention, $\inf \emptyset = \infty$. 
Let $\mu$ be some fixed distribution and $\big(\{ Y_t^{(i)} \}_{t \ge 0}, \tourLength^{(i)}\big)_{i=0}^\infty$ be i.i.d realisations of $(\{ Y_t \}_{t \ge 0}, \tourLength)$ with $Y_0 \sim \mu$. The regeneration times are $T_0 = 0$ and $T_j = \sum_{i=0}^{j-1} \tourLength^{(i)}$ for $j = 1, 2, \dots$. Then the Restore process $\{ X_t \}_{t \ge 0}$ is given by:
\begin{equation*}
    X_t = \sum_{i=0}^\infty \mathbbm{1}_{[T_i, T_{i+1})}(t)Y_{t-T_i}^{(i)}.
\end{equation*}
Let $L_Y$ be the \emph{infinitesimal generator} of $\{Y_t\}_{t \ge 0}$. Then the (formal) infinitesimal generator of $\{X_t\}_{t \ge 0}$ is: $L_X f(x) = L_Y f(x) + \kappa(x) \int [f(y) - f(x)]\mu(y) \dif y$. To use the Restore process for Monte Carlo integration one chooses $\kappa$ so that $\{ X_t \}_{t \ge 0}$ is $\pi$-invariant. Defining $\kappa$ as
\begin{equation} \label{eq:regen_rate_gen_underlying}
    \kappa(x) = \frac{L_Y^\dag \pi(x)}{\pi(x)} + C \frac{\mu(x)}{\pi(x)},
\end{equation}
with $L_Y^\dag$ denoting the formal adjoint, it can be shown that $\int_{\Rd} L_X f(x) \pi(x) \dif x = 0$. Hence, $\{ X_t \}_{t \ge 0}$ is $\pi$-invariant. We will write equation \eqref{eq:regen_rate_gen_underlying} as
\begin{equation} \label{eq:full_regen_rate}
    \kappa(x) = \tilde{\kappa}(x) + C \frac{\mu(x)}{\pi(x)}.
\end{equation}
We call $\tilde{\kappa}$ the \emph{partial regeneration rate}, \correction{$\mu$ the \textit{regeneration distribution}}, $C > 0$ the \emph{regeneration constant}, and $C\mu$ the \emph{regeneration measure}, which must be large enough so that $\kappa(x) > 0, \forall x \in \Rd$. The resulting Monte Carlo method is called the Restore Sampler. Given $\pi$-invariance of $\{ X_t \}_{t \ge 0}$, due to the regenerative structure of the process, we have
\begin{equation*}
    \mathbbm{E}_\pi[f] = \frac{\mathbbm{E}_{X_0 \sim \mu} \Big[\int_0^{\tourLength^{(0)}} f(X_s) \dif s \Big] }{\mathbbm{E}_{X_0 \sim \mu}[\tourLength^{(0)}]}
\end{equation*}
and almost sure convergence of the ergodic averages: as $t \rightarrow \infty$,
\begin{equation} \label{eq:rstr_erg_avgs}
    \frac{1}{t} \int_0^t f(X_s) \dif s \rightarrow \mathbbm{E}_\pi[f].
\end{equation}
For $i=0,1,\dots$, define $Z_i := \int_{T_i}^{T_{i+1}} f(X_s) \dif s$. The Central Limit Theorem for Restore processes states that
\begin{equation*}
    \sqrt{n}\left( \frac{\int_0^{T_n} f(X_s) \dif s}{T_n} - \mathbbm{E}_\pi[f] \right) \rightarrow \mathcal{N}(0, \sigma_f^2),
\end{equation*}
where convergence is in distribution and
\begin{equation} \label{eq:asymp_var}
    \sigma_f^2 := \frac{\mathbbm{E}_{X_0 \sim \mu}\left[\left(Z_0 - \tourLength^{(0)}\mathbbm{E}_\pi[f]\right)^2\right]}{\left(\mathbbm{E}_{X_0 \sim \mu}[\tourLength^{(0)}]\right)^2}.
\end{equation}
Evidently the estimator's variance depends on the expected tour length. This is one motivation for choosing $\mu$ so that tours are on average reasonably long. Indeed, this is the key motivation behind the \textit{minimal} regeneration measure described in Section~\ref{subsec:min_regen}.

%%%%%%%%%%%%%%%%%%%%%%%%%%%%%%%%%%%%%
% Estimating the Normalizing Constant
%%%%%%%%%%%%%%%%%%%%%%%%%%%%%%%%%%%%%

\subsection{Estimating the normalizing constant} \label{sec:est_norm_const}

It is possible to use a Restore process to estimate the normalizing constant of an unnormalized target distribution. When the target distribution is the posterior of some statistical model, its normalizing constant is the \emph{marginal likelihood}, also known as the \emph{evidence}. Computing the evidence allows for model comparison via Bayes factors \citep{kass1995bayes}. Standard MCMC methods draw dependent samples from $\pi$ but cannot be used to calculate the evidence. Alternative methods such as importance sampling, thermodynamic integration \citep{neal1993probabilistic, ogata1989monte}, Sequential Monte Carlo \citep{delmoral2006sequential} or nested sampling \citep{skilling2006nested} must instead be used for computing the evidence \citep{gelman1998simulating}.

For $Z$ the normalizing constant, let
\begin{equation*}
    \pi(x) = \frac{\tilde{\pi}(x)}{Z}.
\end{equation*}
Suppose we are able to evaluate to $\tilde{\pi}$, but $Z$ is unknown. Let the \emph{energy} be defined as:
\begin{equation*}
    U(x) := - \log \pi(x) = \log Z - \log \tilde{\pi}(x).
\end{equation*}
We will see that when $\{ Y_t \}_{t \ge 0}$ is a Brownian motion, $\tilde{\kappa}$ is a function of $\nabla U(x)$ and $\Delta U(x)$, so doesn't depend on $Z$. In the expression for the regeneration rate, the normalizing constant may be ``absorbed'' into $C$. That is,
\begin{equation*}
    \kappa(x) = \tilde{\kappa}(x) + C\frac{\mu(x)}{ \left( \frac{\tilde{\pi}(x)}{Z} \right) } = \tilde{\kappa}(x) + CZ \frac{\mu(x)}{\tilde{\pi}(x)} = \tilde{\kappa}(x) + \tilde{C} \frac{\mu(x)}{\tilde{\pi}(x)},
\end{equation*}
where $\tilde{C} = CZ$. It is known that $C = 1/\mathbbm{E}_\mu[\tourLength]$ \citep[Proof of Theorem 16]{Wang2021}. Since $\tilde{C}$ is set by the user, we have $Z = \tilde{C} \mathbbm{E}_\mu[\tourLength]$. Suppose $n$ tours take simulation time $T$, then \correction{an unbiased} Monte Carlo approximation of $Z$ is:
\begin{equation} \label{eq:Z_monte_carlo_est}
    Z \approx \frac{\tilde{C} T}{n}.
\end{equation}
In Section \ref{sec:adaptive_restore}, we will see that the ability to estimate $Z$ is lost when using adaptive Restore instead of standard Restore, unless the regeneration measure is fixed for that purpose over a sufficient number of iterations.

%%%%%%%%%%%%%%%%%%%%%%
% Minimal Regeneration
%%%%%%%%%%%%%%%%%%%%%%

\subsection{Minimal regeneration}
\label{subsec:min_regen}

The \emph{minimal regeneration measure}, which we denote by $\CMinimum \muMin$, is the choice of $C \mu$ corresponding to the rate being as small as possible:
\begin{equation} \label{eq:min_regen_rate}
    \pkap(x) := \tilde{\kappa}(x) \vee 0 = \tilde{\kappa}(x) + \CMinimum \frac{\muMin(x)}{\pi(x)}.
\end{equation}
We call $\CMinimum$ the \emph{minimal regeneration constant}, $\muMin$ the \emph{minimal regeneration distribution} \correction{and $\pkap$ the \emph{minimal regeneration rate}}. For any $C \mu$ such that $\kappa$, with form \eqref{eq:full_regen_rate}, satisfies $\kappa(x) \ge 0, \forall x \in \Rd$, we have $\pkap(x) \le \kappa(x), \forall x \in \Rd$. Rearranging \eqref{eq:min_regen_rate}, we can obtain an explicit representation for $\muMin$, namely,
\begin{equation}
    \muMin(x) = \frac{1}{\CMinimum}[0 \vee -\tilde{\kappa}(x)]\pi(x).
    \label{eq:mu*}
\end{equation}
A Restore process under $\muMin, \pkap, \CMinimum$ will be referred to as a \textit{minimal restore process}, or simply \emph{minimal restore}. Note that the corresponding notation used by \cite{Wang2021} is $\mu^*, \kappa^*, C^*$. In Section D of the supplementary material \citep{mckimm2024sampling}, a stylised example is considered in a formal asymptotic analysis for large dimensions where we demonstrate that the minimal regeneration distribution leads to $O(1)$ convergence time per unit regeneration.

Frequent regeneration in itself is not necessarily detrimental. For instance, if $\mu \equiv \pi$ and $C$ was large, regeneration would happen very often, but each time the process would start again with distribution $\pi$. Frequent regeneration is more of an issue when $\mu$ is not well-aligned to $\pi$, since the process may then regenerate into areas where $\pi$ has low probability mass, wasting computation.

%A further benefit of minimal restore is that it minimizes the asymptotic variance (in the number of tours) of estimators of $\pi[f]$. This follows from \eqref{eq:asymp_var}, since the expected tour length is maximized.

%%%%%%%%%%%%%%%%%%%%%%%%%%%%%%%%%%%%%%%
% Regeneration-enriched Brownian motion
%%%%%%%%%%%%%%%%%%%%%%%%%%%%%%%%%%%%%%%

\subsection{Regeneration-enriched Brownian motion}\label{subsec:r-BM}

When $\{ Y_t \}_{t \ge 0}$ is a Brownian motion, the partial regeneration rate is
\begin{equation} \label{eq:tilde_kap_BM}
    \tilde{\kappa}(x) := \frac{1}{2} \left( ||\nabla U(x)||^2 - \Delta U(x) \right).
\end{equation}
Regeneration-enriched Brownian motion is the focus of the methodology developed in this article. As such, this subsection is devoted to important aspects of its application to Monte Carlo.

%%%%%%%%
% Output
%%%%%%%%

\subsubsection{Output}

The left-hand side of equation \eqref{eq:rstr_erg_avgs} cannot be evaluated exactly when the underlying process is a Brownian motion. Instead, the output of the sampler is the state of the process either at fixed, evenly-spaced intervals or at the arrival times of an exogenous, homogeneous Poisson process with rate $\outputRate$. We use a homogeneous Poisson process to record output events, since this method is marginally simpler -- see the discussion in Section
A of the supplementary material \citep{mckimm2024sampling}.

%%%%%%%%%%%%
% Simulation
%%%%%%%%%%%%

\subsubsection{Simulation}

\emph{Poisson thinning} \citep{Lewis1979} is used to simulate regeneration events. This is because the regeneration rate is itself a stochastic process, given by $t\mapsto \kappa_t := \kappa(X_t),$ and hence no closed form expression for the right-hand side of \eqref{eq:tour_length} is available. 
Suppose $\kappa$ is uniformly bounded. That is,
\begin{equation*}
    \K := \sup_{x \in \Rd} \kappa(x)<\infty .
\end{equation*}
Then $\kappa_t < \K, \forall t \ge 0$ and $\K$ may be used as the dominating rate in Poisson thinning. To simulate a rate $\kappa_t$ Poisson process, at time $t$ generate $\timeToNextPotentialRegen \sim \text{Exp}(\K)$, the time to the next potential regeneration event. Then regenerate at time $t + \timeToNextPotentialRegen$ with probability $\kappa_{t + \timeToNextPotentialRegen} / \K$, else don't regenerate. In \ref{sec:min_bmr} we consider the process simulated using the minimal rate $\pkap$ given in \eqref{eq:min_regen_rate}. In this case, let
\begin{equation*}
    \pK := \sup_{x \in \Rd} \pkap (x).
\end{equation*}

Algorithm~\ref{algo:bmrstr} shows how to simulate a Brownian motion Restore process for a fixed number of tours.
%Dominating rate $K$ is used, but this will be $\regenRateTrunc$ when $\kappa$ is unbounded.
%The process is initialized with a regeneration, by simulating from $\mu$.
Variables $X, t$ and $i$ denote the current state, time and tour number of the process.
%At each iteration of the while loop, the times to the next potential regeneration and output events, $\timeToNextPotentialRegen$ and $s$ are simulated. If the next output event arrives before the next potential regeneration event, the process is simulated forward to the next output time and its state, time and tour number are recorded. Otherwise, the process is updated forward in time to the next potential regeneration event, and is simulated whether regeneration occurs or not and then the current state, time and tour are updated.

\begin{algorithm}[H]
    $X \sim \mu, t \gets 0, i \gets 0$\\
    \While{$i < n$}{
        $\timeToNextPotentialRegen \sim \text{Exp}(\K), s \sim \text{Exp}(\outputRate)$\\
        \eIf{$s < \timeToNextPotentialRegen$}{
            $t \gets t + s, X \sim \mathcal{N}(X, s)$. Record $X, t, i$
        }{
            $t \gets t + \timeToNextPotentialRegen, X \sim \mathcal{N}(X, \timeToNextPotentialRegen), u \sim \mathcal{U}[0,1]$\\
            \If{$u < \kappa(X) / \K$}{
                $X \sim \mu, i \gets i + 1$
            }
        }
    }
    \caption{Brownian motion Restore}
    \label{algo:bmrstr}
\end{algorithm}

For many target distributions $\kappa$ is \textit{not} bounded. Then, to use a global dominating rate, $\kappa$ must be \emph{truncated} at some level.
Alternatively, when $\kappa$ is bounded but the bound is very large, then for simulation purposes truncation may be desirable. When a truncated regeneration rate is used, we will denote the truncation level as $\regenRateTrunc$. When a truncated minimal regeneration rate is used, the truncation level will be denoted $\minRegenRateTrunc$. In other words, this notation signals the use of rates
\begin{equation*}
    \kappa(x) = \left( \tilde{\kappa}(x) + C \frac{\mu(x)}{\pi(x)} \right) \wedge \regenRateTrunc
\end{equation*}
and
\begin{equation*}
    \kappa(x) = \pkap(x) \wedge \minRegenRateTrunc.
\end{equation*}

Truncation introduces error, so that the Monte Carlo approximation is no longer exact, however the error is negligible for $\regenRateTrunc$ large enough. Indeed, in theory it is possible to quantify the size of this error \citep[Theorem 30]{Rudolf2021,Wang2021} and show that as $\regenRateTrunc$ goes to infinity, the error tends to zero \citep[Proposition 32]{Wang2021}. Bounding the error explicitly using Theorem~30 of \citep{Wang2021} may be difficult in challenging problems. For complicated posterior distributions, it may not be possible to compute the supremum of $\kappa$, in which case one cannot be sure that the global dominating rate does not truncate $\kappa$. An advantage of using the minimal rate is that for a given error tolerance, the truncation level $\minRegenRateTrunc$ typically only needs to be increased logarithmically with dimension $d$; see Section~\ref{sec:min_bmr}.

%%%%%%%%%%%%%%%%%%%%%%%%%%%%%%%%%
% Minimal Brownian Motion Restore
%%%%%%%%%%%%%%%%%%%%%%%%%%%%%%%%%

\subsubsection{Minimal Brownian motion Restore} \label{sec:min_bmr}

A \emph{minimal Brownian motion Restore process} is defined by enriching an underlying Brownian motion with regenerations from distribution \eqref{eq:mu*} at rate \eqref{eq:min_regen_rate}, with $\tilde{\kappa}$ given by \eqref{eq:tilde_kap_BM}. Figure \ref{fig:min_bmr_path} shows five tours of a minimal Brownian motion Restore process. Green dots and red crosses show the first and last output state of each tour. In this example, $\muMin$ is supported on the interval $[-1,1]$, shown by gray lines. Note that the process always regenerates from outside to inside this interval.

\begin{figure}
    \centering
    \includegraphics[width=.9\linewidth, trim={0 1cm 0 2cm}, clip]{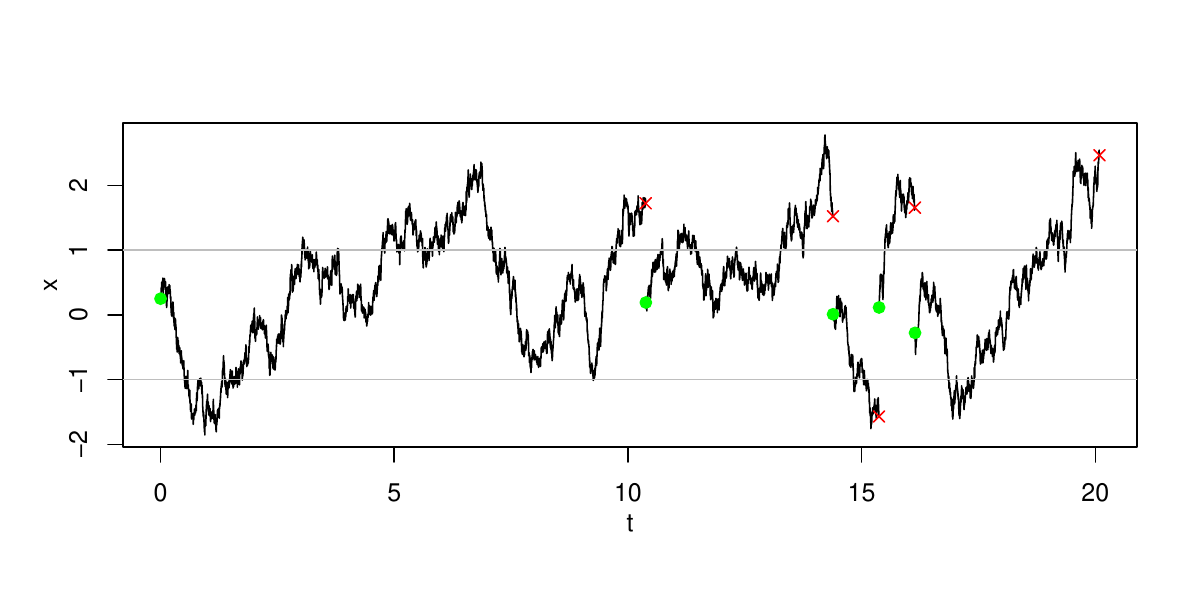}
    \caption{Five tours of a minimal Brownian motion Restore process with $\pi \equiv \mathcal{N}(0,1)$. The first and last output state of each tour is shown by a green dot and red cross respectively. The gray horizontal lines mark the support of $\muMin$.}
    \label{fig:min_bmr_path}
    % 4 by 8 inches
    % path generated by file ~/rwm/uvg/bmr/min_bmrstr_uvg.cpp
    % plot generated by file ~/rwm/uvg/bmr/bmrstr_uvg.R
\end{figure}

An advantage of minimal Brownian motion Restore is that it reduces computational expense. A useful feature of $\pkap$ is that to ensure $\P[\pkap < \minRegenRateTrunc] < 1 - \epsilon$ is satisfied, for $\epsilon > 0$ a small constant (e.g. $\epsilon=0.001$), $\minRegenRateTrunc$ scales logarithmically with dimension $d$. Supplement D provides an illustration of this, when $X \sim \N(0,I)$ \citep{mckimm2024sampling}.

%%%%%%%%%%%%%%%%%%%%%%%%%%%%%%%%%%%%%%%%
% Adapting the Regeneration Distribution
%%%%%%%%%%%%%%%%%%%%%%%%%%%%%%%%%%%%%%%%

\section{Adapting the regeneration distribution} \label{sec:adaptive_restore}

A serious impediment to using a standard Restore process as a sampler is the choice of $C$ and $\mu$. Firstly, for a given $\mu$, $C$ must be large enough that $\kappa(x) \ge 0, \forall x \in \Rd$. Secondly, it is unclear how to set $\mu$ in practice.

When $\mu$ is a poor approximation of $\pi$, the regeneration rate can become so large that simulation is actually not feasible. For example, consider the hierarchical model of pump failures used as a benchmark in the Gibbs sampling literature \citep{carlin1991iterative}: $R_i \sim \text{Poisson}(X'_i t_i)$; $i=1,\dots,10; X'_i \sim \text{Gamma}(c_1, X'_{11}); i=1,\dots,10; X'_{11} \sim \text{InverseGamma}(c_2, c_3)$; with constants $c_1 = 1.802, c_2 = 2.01, c_3 = 1.01$. Observation $R_i ~ (i=1,\dots,10)$ is the number of recorded failures of pump $i$, which is observed for a unit of time $t_i ~ (i=1,\dots,10)$. The failure rate of pump $i$ is $X'_i ~ (i=1,\dots, 10)$. Before sampling, we transformed the posterior to be defined on $\Rd$ by making a change-of-variables, defining $X_i = \log X'_i ~ (i=1,\dots,10)$. We then transformed the posterior again, based on its Laplace approximation, as described in Supplement B \citep{mckimm2024sampling}. Let $\mu \equiv \mathcal{N}(0, I)$, which may seem like a reasonable choice, since the posterior has been transformed to be roughly centred at zero and have an identity covariance matrix. To satisfy $\P(\kappa(X) < \regenRateTrunc) = 0.99$, we would need to set $\regenRateTrunc \approx 6.5 \times 10^{15}$. A computer typically represents numbers with either 7 (for a float) or 15 (for a double) decimal digits of precision. Thus using $\regenRateTrunc \approx 6.5 \times 10^{15}$ is not practical, because the average time between potential regeneration events would be approximately $1.5 \times 10^{-16}$ and there would be significant numerical difficulties in simulating a Brownian motion over this time-span.

The above strategy for choosing $\mu$ was to set it as an approximation of $\pi$. To understand why it resulted in a very large regeneration rate, it is important to understand that $\tilde{\kappa}$ is negative in an area $\cX$ of the state-space $\Rd$, which is close to the mode of $\pi$, and that a role of the term $C \mu / \pi$ in the expression for $\kappa$ is to make $\kappa$ non-negative in $\cX$. If $\pi$ is skewed, then in certain directions $\mu$ may have lighter tails than $\pi$ and in other directions heavier tails. In the directions for which $\mu$ has lighter tails than $\pi$, the ratio $\mu / \pi$ may be small in parts of $\cX$. In this case, $C$ must compensate by being very large, which pushes up $\kappa$ everywhere. This can result in tours of the process frequently starting in areas where $\pi$ has low probability mass and for which the regeneration rate is very large, so regeneration occurs very frequently. This is computationally wasteful, since $\pi$ and its derivatives must be evaluated in order to determine regeneration events. Even if sophisticated methods were used to approximate $\pi$, as the dimension of the problem increases it would becomes harder to find a good approximation, so $\kappa$ would become very large which results in expensive computation.

Consider the two strategies of attempting to choose a distribution $\mu$ with tails that are (i) heavier and (ii) lighter than those of $\pi$. For example, one might choose $\mu$ as (i) a t-distribution or (ii) a Gaussian distribution with covariance matrix $\epsilon I$, for $\epsilon > 0$ a small constant. For both strategies, in practice there is no guarantee that the chosen $\mu$ indeed has tails that are heavier/lighter than those of $\pi$, since before any sampling has occurred the geometry of $\pi$ is unknown. Even if $\mu$ does has heavier tails than $\pi$, though $C$ would not need to compensate for the ratio $\mu / \pi$, the process would still display the behaviour of frequently regenerating into regions of low probability. If $\mu$ has lighter tails than $\pi$, the sampler may not take full advantage of global dynamics as a means of making large moves around the state space.

As mentioned in the introduction, $\mu$ could only be adapted to become closer to $\pi$ if the initial choice of $\mu$ were a `good' one that allows $\kappa$ to be small enough for sampling to be possible. The above example of a model of pump failure demonstrates that it is easy to make an initial `bad' choice of $\mu$.

%%%%%%%%%%%%%%%%%%
% Adaptive Restore
%%%%%%%%%%%%%%%%%%

\subsection{Adaptive Restore}

An adaptive Restore process uses the regeneration rate that is as small as possible. In doing so, it removes the need to find a suitably large constant $C$ and replaces the choice of fixed $\mu$ with a less consequential choice of an initial regeneration distribution. The process is defined by enriching some underlying continuous-time Markov process with regenerations at rate $\pkap$ from, at time $t$, a distribution $\mu_t$. Initially, the regeneration distribution is $\mu_0$. The regeneration distribution is updated by adding point masses at certain time points. Let $\pi_t$ be the stationary distribution of the Restore process with fixed regeneration distribution $\mu_t$. We have simultaneous convergence of $(\mu_t, \pi_t)$ to $(\muMin, \pi)$. The density of $\mu_t$ is given by
\begin{equation}
    \label{eq:mu_t}
    \mu_t(x) =
        \correction{ \frac{1}{a+N(t)} }\sum_{i=1}^{N(t)} \delta_{X_{\zeta_i}}(x) + \correction{ \frac{a}{a+N(t)} }\, \mu_0(x),
\end{equation}
where $a>0$ is some constant, $\mu_0$ is some fixed initial distribution and $\zeta_1, \zeta_2, \dots, \zeta_{N(t)}$ are the arrival times of an inhomogeneous Poisson process $(N(t):t\ge 0)$ with rate $t \mapsto \nkap(X_t)$,
\begin{equation*}
    \nkap(x) := [0 \vee -\tilde{\kappa}(x)].
\end{equation*}
The rate $\nkap$ Poisson process is simulated using Poisson thinning, so it is assumed that there exists a constant
\begin{equation*}
    \nK := \sup_{x \in \mathcal{X}} \nkap(x),
\end{equation*}
such that $\nK > 0$. The distribution $\mu_t$ is a mixture of a fixed distribution $\mu_0$ and a discrete measure $N(t)^{-1}\sum_{i=1}^{N(t)} \delta_{X_{\zeta_i}}(x)$. The constant $a$ is called the \emph{discrete measure dominance count}, since it is the time at which regeneration is more likely to be from the discrete measure in the mixture distribution.

Algorithm \ref{algo:adaptive_restore} describes the method. Three Poisson processes, one homogenous and two inhomogeneous, are simulated in parallel. Here, the process is generated for a fixed number of tours, though another condition such as the number of samples or simulation time could equally be used. 

\begin{algorithm}[H]
    \SetAlgoLined
    $t \gets 0, E \gets \emptyset, i \gets 0, X \sim \mu_0.$\\
    \While{$i < n$}{
        $\timeToNextPotentialRegen \sim \text{Exp}(\pK), s \sim \text{Exp}(\outputRate), \muMinDomEvent \sim \text{Exp}(\nK)$.\\
        \uIf{$\timeToNextPotentialRegen < s$ and $\timeToNextPotentialRegen < \muMinDomEvent$}{
            $X \sim \mathcal{N}(X, \timeToNextPotentialRegen), t \gets t + \timeToNextPotentialRegen, u \sim \mathcal{U}[0,1]$.\\
            \If{$u < \pkap(X)/\pK$}{
                \eIf{$|E| = 0$}{
                    $X \sim \mu_0$.
                }{
                    $X \sim \mathcal{U}(E)$ with probability \correction{ $N(t) / \big(a + N(t) \big)$ }, else $X \sim \mu_0$.
                }
                $i \gets i+1$.
            }
        }
        \uElseIf{$s < \timeToNextPotentialRegen$ {\bf and} $s < \muMinDomEvent$}{
            $X \sim \mathcal{N}(X,s), t \gets t+s$, record $X, t, i$.
        }
        \Else{
            $X \sim \mathcal{N}(X, \muMinDomEvent), t \gets t + \muMinDomEvent, u \sim \mathcal{U}[0,1]$.\\
            If $u < \nkap(X)/\nK$ then $E \gets E \cup \{ X \}$.
        }
    }
    \caption{Adaptive Brownian motion Restore}
    \label{algo:adaptive_restore}
\end{algorithm}

Figure \ref{fig:cabmr_uvg} shows the path of an adaptive Restore process with $\pi \equiv \mathcal{N}(0,1), \mu_0 \equiv \mathcal{N}(2,1)$ and $a=10$. The regeneration rate $\pkap$ encourages the process to drift to towards the origin, where the target distribution is centred. To allow convergence of the process, one might want to only record output after some \emph{burn-in} time $b$.

\begin{figure}
    \centering
    \includegraphics[width=\linewidth, trim={0 1cm 0 1cm}, clip]{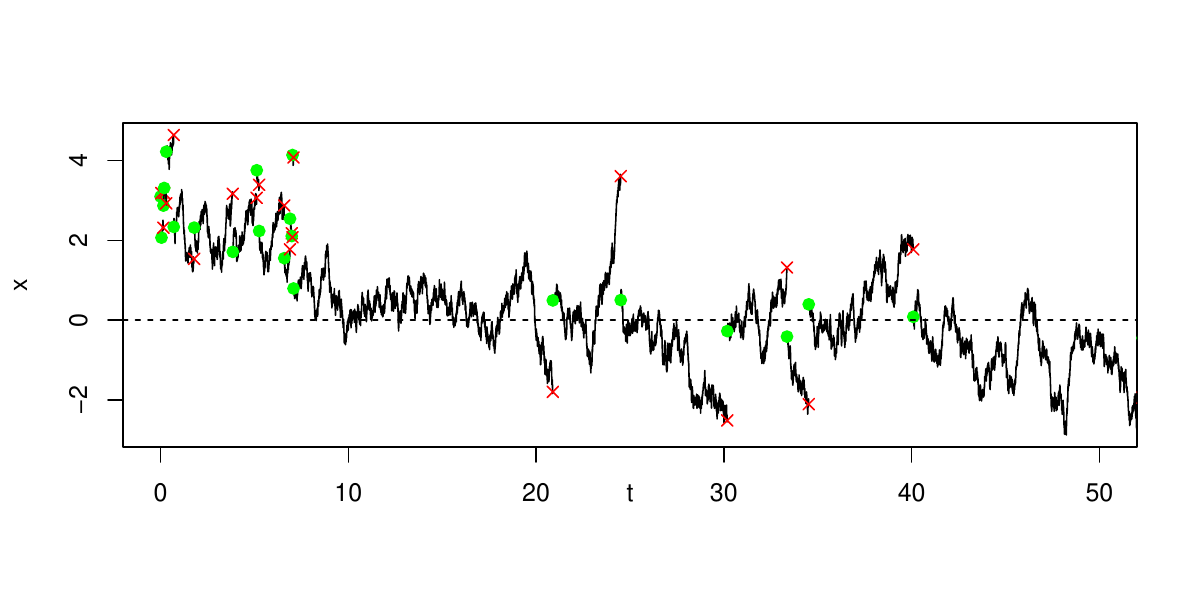}
    \caption{Path of an adaptive Restore process with $\pi \equiv \mathcal{N}(0,1), \mu_0 \equiv \mathcal{N}(2,1), a=10$. Green dots and red crosses show the first and last output states of each tour.}
    \label{fig:cabmr_uvg}
\end{figure}

Note that for adaptive Restore, one can not straightforwardly apply the method described in \ref{sec:est_norm_const} in order to estimate the normalizing constant $Z$. This is because we do not explicitly set constant $\tilde{C}$ and hence cannot use equation \eqref{eq:Z_monte_carlo_est}.

%%%%%%%%%%%%%%%%%%%%%%%%%%%%%%%%%%%%%%%%%%%%%%%%%%%%%%%%%%%%
% Choice of initial regeneration distribution and parameters
%%%%%%%%%%%%%%%%%%%%%%%%%%%%%%%%%%%%%%%%%%%%%%%%%%%%%%%%%%%%

\subsection{Choice of initial regeneration distribution and parameters}

Generally, we set $\mu_0$ to approximate $\pi$, e.g., $\mu_0 \equiv \mathcal{N}(0, I)$ ($\pi$ undergoes a pre-transformation based on its Laplace approximation, as described in Section B of the supplementary material \citep{mckimm2024sampling}, so that the transformed $\pi$ is approximately $\mathcal{N}(0, I)$). Setting $\mu_0$ as a more sophisticated approximation of $\muMin$ might lead to faster converge, but for the example problems considered, this simpler choice of $\mu_0$ suffices.

For $\muMin_{\N}$ the minimal regeneration distribution of an isotropic Gaussian distribution, we experiment with using $\mu_0 \equiv \muMin_{\N}$ for an adaptive Restore process. There is no guarantee that the support of $\muMin_{\N}$ will cover the support of $\muMin$, so $\muMin_{\N}$ can not be used for standard Restore sampling.

There is a tradeoff in choosing $a$, the discrete measure dominance count. Empirical experiments have shown that smaller choices of $a$ can lead to faster convergence. However, a larger value of $a$ encourages more regenerations from $\mu_0$, which makes it more likely for $\{ X_t \}_{t \ge 0}$ to explore regions it has not previously visited.

For this paper's examples, $\minRegenRateTrunc$ and $\nK$ were selected based on the quantile functions of $\tilde{\kappa}$, approximated using the output of a preliminary Markov chain generated with a Metropolis-Hastings algorithm. It may be possible to learn suitable values of $\minRegenRateTrunc$ and $\nK$ on-the-fly, without the need for a preliminary run of a Markov chain. Assuming $\pi$ is $d$-dimensional and close to Gaussian, a sensible initial guess of $\nK$ is $d/2$, since if $\pi \equiv \N(0, I)$ then $\nK=d/2$ exactly. This initial guess of $\nK$ could then be adjusted as necessary. Similarly, a sensible initial estimate for $\minRegenRateTrunc$ could be made based on the cumulative distribution function of a chi-squared random variable, see Supplement D \citep{mckimm2024sampling}, then adjusted by monitoring how often $\kappa$ exceeds this truncation level.

%%%%%%%%%%%%%%%%%%%%%%%%%%%%%%%%%%%%%%%%%
% Adaptive Restore with Short-Term Memory
%%%%%%%%%%%%%%%%%%%%%%%%%%%%%%%%%%%%%%%%%

\subsection{Adaptive Restore with Short-Term Memory}

\correction{A practical issue with simulating an adaptive Restore process is that it may require a lot of memory. In fact, we have found that the memory available on a core of a server may be exceeded before the process has converged. An option for controlling the memory requirements of the algorithm is to simulate a process with ``Short-Term Memory''. As the process is simulated, particles are added to the regeneration distribution as usual, but particles are also removed on a "first in first out" basis: the particles that were added first are the first to be forgotten. In programming terms, the cloud of particles is a queue: particles are added to the back and removed from the front. Empirically, we have observed that this modification has the additional benefit of improving the speed of convergence of the process.}

\correction{Let an adaptive Restore process with Short-Term Memory be designed so that until $\nCloud$ particles have been added to the regeneration distribution, no particles are forgotten. Then, once this threshold is reached, for every $\nForget$ particles that are added, $\nForget - 1$ are forgotten, so that the number of particles in the regeneration distribution only increases by 1. The regeneration distribution at time $t$ is similar to equation \eqref{eq:mu_t} and given by:
\begin{equation*}
    \mu_t(x) =
        \frac{1}{a + N(t) - l(t) + 1} \sum_{i=l(t)}^{N(t)} \delta_{X_{\zeta_i}}(x) + \frac{a}{a + N(t) - l(t) + 1}\, \mu_0(x), 
\end{equation*}
where $l(t) = \lfloor (N(t) - \nCloud)(\nForget-1)/\nForget \rfloor\vee 0$.
}

%%%%%%%%%%%%%%%%%%%%%%%%%%
% Connections with ReScaLE
%%%%%%%%%%%%%%%%%%%%%%%%%%

\subsection{Connections with ReScaLE}

Although inspired by the Restore algorithm of \cite{Wang2021}, the adaptive Restore algorithm presented in Algorithm~\ref{algo:adaptive_restore} has many connections to \emph{quasi-stationary Monte Carlo} (QSMC) methods such as the ScaLE Algorithm of \cite{Pollock2020}, and particularly the ReScaLE algorithm of \cite{Kumar2019}.

Indeed, ScaLE and ReScaLE can be seen as instances of Restore where the regeneration distribution is chosen to \textit{be} the target, which is then learnt adaptively: the killing rate for QSMC methods is given by $\kappa = \tilde \kappa + C$,
% change made to cut down the length of the paper
and in the case of ReScaLE, at a killing time $T$, the process is regenerated from its empirical occupation measure:
\begin{equation}
    X_T \sim \frac{1}{T}\int_0^T \delta_{X_s}\dif s.
    \label{eq:rescale_empirical}
\end{equation}
The key motivation behind ScaLE and ReScaLE was applicability to \textit{tall data problems}, due to the applicability of exact subsampling techniques \cite{Pollock2020, Kumar2019}, however sampling from \eqref{eq:rescale_empirical} is somewhat delicate due to the need to simulate complex diffusion bridges.

By contrast, although adaptive Restore does not straightforwardly permit exact subsampling, its regeneration mechanism is considerably simpler to implement, and it is only required to adaptively learn the compactly-supported distribution $\muMin$. ReScaLE need learn the entire distribution $\pi$---approximated by the trajectory of the diffusion path -- on $\Rd$ for its regeneration mechanism, and thus the two algorithms---although similar in many regards---can be seen as complementary.

\correction{Finally, although it is possible to use more exotic choices of underlying local dynamics than a Brownian motion, as in \cite{Wang2019TheoQSMC,Wang2021}, it has been shown for QSMC methods that choosing Brownian motion -- which is of course mathematically and computationally convenient -- can lead to a maximal spectral gap; see \cite[Section~2]{Wang2019TheoQSMC}. This further justifies our use of regeneration-enriched Brownian motion as in Section~\ref{subsec:r-BM}.}

%%%%%%%%
% Theory
%%%%%%%%

\section{Theory} \label{sec:theory}
In this section we will establish the validity of the adaptive Restore algorithm, as described in Algorithm~\ref{algo:adaptive_restore}. In general, this is a difficult task since the process is self-reinforcing, on a noncompact state space; most works in the literature are for compact state spaces, 
%\cite{Aldous1988, Benaim2018, Wang2020}, 
an exception being \cite{Mailler2020}. In our present setting, we will thus establish validity by showing that the measures $\mu_t$ in \eqref{eq:mu_t} converges weakly almost surely to the minimal regeneration distribution $\muMin$ as in \eqref{eq:mu*}, which ultimately implies validity of the adaptive Restore algorithm. For this theoretical analysis, we consider a fixed regeneration rate; in particular we do not consider questions related to truncating a possibly divergent regeneration rate.

This section is self-contained, to be illustrated by numerical experiments in Section~\ref{sec:examples}.

%\subsection{Basic notation}

%$\mathcal P(\cX)$ denotes the space of probability measures on a measurable space $\cX$. \gareth{Remove notation $\mathcal{P}(\mathcal{X})$.}

%%%%%%%%%%%%%%%%%%%%%%%%%%%%%%%%%%%%%%
% Summary of theoretical contributions
%%%%%%%%%%%%%%%%%%%%%%%%%%%%%%%%%%%%%%

\subsection{Summary of theoretical contributions and related work}
\label{subsec:summary_theory}
The theoretical analysis in this section is based on stochastic approximation techniques, following a similar overall approach as in the sequence of previous works already cited.
%\cite{Aldous1988,Benaim2015a, Blanchet2016, Benaim2018, Wang2020, Mailler2020}. 
Our proof most closely follows the approaches of \cite{Benaim2018,Wang2020}, with several key novelties. The former
%article of \cite{Benaim2018} 
shows almost sure convergence of stochastic approximation algorithms for \textit{discrete-time} processes on \textit{compact spaces}. By focussing our attention on the measures $\mu_t$, which in our present setting are supported on a compact set, we will be able to identify an appropriate \textit{embedded discrete-time Markov chain}, to which we can apply their main results.
%of \cite{Benaim2018}. 
Thus the main technical work of this section is identifying the appropriate discrete-time structure, and then checking that the relevant hypotheses are satisfied.
%in order to utilise \cite{Benaim2018}.

A further difference between our present analysis and the previous works cited above concerns the nature of the killing mechanism. In all previous works, the killing mechanism was given by an additional random clock $\tau_\partial$ of the form
\begin{equation}
    \tau_\partial := \inf \left\{ s\ge 0: \int_0^s \kappa(X_u) \dif u\ge \xi \right\},
    \label{eq:general-tau}
\end{equation}
for an appropriate killing rate $\kappa:\mathcal X \to [0,\infty]$ and $\xi\sim \mathrm{Exp}(1)$ independent (in discrete-time settings the obvious modifications need to be made).

By contrast, our present setting is considering \textit{two competing clocks} $\zeta$ and $T$, each of which defined as in \eqref{eq:general-tau} with their own respective arrival rates $\kappa^\pm$ and independent exponential random variables $\xi, \xi'$. A `killing' event in our setting is then the event $\zeta < T$,
namely that the clock with rate $\kappa^-$ rings before the clock with rate $\kappa^+$.

Even with these key differences, we show in this section that the general stochastic approximation  approach of \cite{Benaim2018,Wang2020} can still be applied, and thus we deduce almost sure weak convergence of the measures $\mu_t$.

%%%%%%%%%%%%%%%%%%%%%%%%%%%%%%%%%%%%%%%
% Diffusion setting and Restore process
%%%%%%%%%%%%%%%%%%%%%%%%%%%%%%%%%%%%%%%

\subsection{Diffusion setting and Restore process}

We assume that we are given some underlying local dynamics on the Euclidean space $\R^d$, with generator $L$, assumed to be a self-adjoint (reversible) diffusion:
\begin{equation}
    \dif X_t = \nabla A(X_t)\dif t + \dif W_t.
    \label{eq:SDE}
\end{equation}
This is a symmetric diffusion, with a self-adjoint generator $L$ on the Hilbert space $\mathcal L^2(\Gamma)$, where 
\begin{equation*}
    \Gamma(\dif y) = \gamma(y)\dif y,\quad \gamma(y) = \exp(2A(y)),\quad \forall y \in \R^d.
\end{equation*}
%For Brownian motion, $\Gamma$ reduces to Lebesgue measure, and we often work in that case for notational simplicity; but the analysis extends to the more general reversible case.

We have fixed a target density $\pi$ on $\R^d$ \correction{with respect to Lebesgue measure}. We then define a \textit{partial killing rate} $\tilde \kappa: \R^d \to \R$, which comes from \cite{Wang2019TheoQSMC}: 
\begin{equation*}
    \tilde \kappa (y) := \frac{1}{2}\left( \frac{\Delta \pi}{\pi} -\frac{2\nabla A \cdot \nabla \pi}{\pi} -2\Delta A\right)(y),\quad y\in \R^d.
\end{equation*}
In the special case of Brownian motion ($A\equiv 0$), this reduces to \eqref{eq:tilde_kap_BM}.

We then define the positive and negative parts:
\begin{equation*}
    %\begin{split}
        \kappa^+ := \tilde \kappa \vee 0,\qquad
        \kappa^- := -(\tilde \kappa \wedge 0).
    %\end{split}
\end{equation*}
We make the following regularity assumptions (\textit{c.f.} \cite{Wang2019TheoQSMC}).
\begin{assumption}
    $A$ is smooth ($C^\infty$), such that the SDE \eqref{eq:SDE} has a unique nonexplosive weak solution. The target density $\pi$ is smooth and positive, and such that $\int \pi^2/\exp(2A)\dif y<\infty$. Thus $\tilde \kappa$ is continuous, which implies that the functions $\kappa^+, \kappa^-$ are continuous.
    Furthermore, $\kappa^- \le \nK$ uniformly for some $\nK>0$.
    \label{assm:regularity}
\end{assumption}
The $C^\infty$ assumption on $A$ may appear strong, but is needed for technical reasons in some proofs (e.g. \cite[Lemma 15]{Wang2021}). Of course, the most common chosen choices $A\equiv 0$ or $A$ quadratic are $C^\infty$.

\begin{assumption}
    The support of $\kappa^-$ is bounded: the set
    \begin{equation*}
        \mathcal X := \{x\in \R^d: \tilde \kappa(x)\le 0\}
    \end{equation*}
    is a compact subset of $\R^d$.
    \label{assm:compactX}
\end{assumption}
\begin{remark} When we use Brownian motion as local dynamics, this is a weak condition, holding for instance when $\pi$ satisfies a suitable sub-exponential tail condition \citep{Wang2019TheoQSMC}.
%    This assumption holds true for all practical examples for Brownian motion; a sufficient condition in this case is that $\pi$ possesses a sub-exponential tail; see the remarks in \cite{Wang2019TheoQSMC}.
\end{remark}
\correction{Henceforth, to reduce the notational burden and to aid readability, we assume that we are using Brownian motion for local dynamics, so $\Gamma$ reduces to Lebesgue measure. However the analysis extends to the more general reversible case.}

Thus, the sub-Markov semigroup corresponding to the diffusion $X$, killed at rate $\kappa^+$ can also be realised as a self-adjoint semigroup on $\mathcal L^2(\Gamma)$. 
Furthermore, there exists a transition sub-density $p^{\kappa^+}(t,x,y)$, as in \cite{Kolb2012}, following from the derivation of \cite{Demuth2000}: writing $T_1$ for the first killing event,
\begin{equation}
    \E_x[f(X_t) \mathbb I_{t < T_1}] = \int p^{\kappa^+} (t,x,y)f(y) \,\dif y.
    \label{eq:subdensity}
\end{equation}
From \cite{Demuth2000}, the function $(t,x,y)\mapsto p^{\kappa^+} (t,x,y)$ will be jointly continuous and symmetric in $x, y$.

\begin{assumption}
    The killing time $T_1$ has uniformly bounded expectation on $\cX$: $\sup_{x\in \cX}\E_x[T_1]<\infty.$
% \todo{NOTE: we maybe only need this on $\cX$ in which case it holds automatically.}
    \label{assm:finiteT1}
\end{assumption}
\begin{remark}
    This is a very weak assumption, since $\cX$ is a compact set. A much stronger condition will be satisfied---uniform bounded expectation over the \textit{entirety} of $\Rd$ for Brownian motion---if the killing rate satisfies
    \begin{equation*}
        \liminf_{\|x\|\to\infty }\tilde \kappa(x)>0,
    \end{equation*}
    which is the case when $\pi$ possesses a sub-exponential tail; see \cite{Wang2019TheoQSMC}.
\end{remark}
We shall also require the following elementary identity.
\begin{lemma}
    For any $x\in \cX$, we have the identity: $\E_x[T_1] = \int_0^\infty \dif t\int \dif y\,  p^{\kappa^+} (t,x,y).$
    \label{lem:ET1}
\end{lemma}
\begin{proof}
    We have $\int_0^\infty \dif t\int_{\cX} \dif y\,  p^{\kappa^+} (t,x,y) = \int_0^\infty \dif t\, \P_x(T_1 \ge t) = \E_x[T_1]$.
\end{proof}

%%%%%%%%%%%%%%%%%
% Restore process
%%%%%%%%%%%%%%%%%

\subsection{Restore process}

Recall the \textit{minimal regeneration distribution}, as in \eqref{eq:mu*}, which has a density function with respect to Lebesgue measure on $\Rd$, compactly supported on $\cX$, given by
\begin{equation*}
    \muMin( y) := \frac{1}{\CMinimum} \pi(y)\kappa^-(y) ,
\end{equation*}
where $\CMinimum := \int\pi(y)\kappa^-(y) \dif y $ is the normalizing constant.

We consider now running a Restore process $X$ with local dynamics \eqref{eq:SDE} described by infinitesimal generator $L$, regeneration rate $\kappa^+$ and regeneration distribution $\mu$. If $\mu=\muMin$, then $\pi$ will be the invariant distribution of $X$ under appropriate regularity conditions; see \cite{Wang2021}.

The goal of the adaptive Restore algorithm is to \textit{learn} $\muMin$ adaptively, by running an additional Poisson process with rate function $t \mapsto \kappa^-(X_t)$. These auxiliary arrival times will be used to construct the adaptive estimate of $\muMin$. Notationally, we will use letters $T,T_1, T_2, \dots$ to refer to regeneration times of the Restore process, which arrive with rate $\kappa^+(X_t)$, and $\muMinEvent, \muMinEvent_1, \muMinEvent_2, \dots$  to refer to the addition events which arrive with rate $\kappa^-(X_t)$. In particular, for a Restore process with local dynamics $L$, regeneration rate $\kappa^+$ and regeneration distribution $\mu$---abbreviated into Restore($L, \pkap, \mu)$---we have, $X_{T_i}\sim \mu$.

We have the following representation from \cite{Wang2021} of the invariant measure of Restore$(L,\kappa, \mu)$:
\begin{equation*}
    \Pi_{\kappa}(\mu)(\dif y) \propto \int \mu(\dif x)\int_0^\infty \dif t \,p^{\kappa}(t,x,y) \dif y.
\end{equation*}
In particular, we must therefore have the identity
\begin{equation}
    \pi(y) \propto \int \muMin(\dif x)\int_0^\infty \dif t\, p^{\kappa^+} (t,x,y).
    \label{eq:pi_mu*_rep}
\end{equation}

%%%%%%%%%%%%%%%%%%%%%%
% Discrete-time system
%%%%%%%%%%%%%%%%%%%%%%

\subsection{Discrete-time system}

For now we imagine the rebirth distribution to be a fixed (but arbitrary) measure $\mu$ and consider the Restore($L, \pkap, \mu)$ process, with additional events $(\muMinEvent_i)_{i=1}^\infty$ at rate $\kappa^-(X_t)$. We will let $\E_x$ denote the expectation under the law of this Restore process $X$ initialised from $X_0 = x$. We are interested in studying the behaviour of the points $(X_{\muMinEvent_1}, X_{\muMinEvent_2},\dots)$.

Our first goal is to show that this sequence is in fact a Markov chain, and give an expression for its transition kernel. To reduce notational clutter, we will write $\muMinEvent:= \muMinEvent_1$, and $T:=T_1$ for the first $\kappa^-$ or regeneration events respectively:
\begin{equation*}
    \begin{split}
        \muMinEvent &:= \inf \left\{ s\ge 0: \int_0^s \kappa^-(X_u) \dif u\ge \xi \right\},\\
        T &:= \inf \left\{ s\ge 0: \int_0^s \kappa^+(X_u) \dif u\ge \xi' \right\},
    \end{split}
\end{equation*}
where $\xi,\xi' \sim \mathrm{Exp}(1)$ are independent of each other and of all other random variables.

\begin{lemma} \label{lemma:Y_markov_chain}
    Defining for each $i\in \mathbb N$, $Y_i := X_{\muMinEvent_i}$, the sequence $(Y_i)_{i=1}^\infty$ is a Markov chain on $\cX$.
    \label{lem:Y_isMC}
\end{lemma}
\begin{proof}
    This follows from the fact that the underlying Restore process $X$ is a strong Markov process, and the fact that the Poisson processes have independent exponential random variables.
\end{proof}

We now define a Markov \textit{sub}-kernel $\subKernel (x,\dif y)$ on $\cX$ which will be crucial to describing the transition kernel of the chain $Y$. The kernel is defined, for any integrable $f:\cX \to \R$, by
\begin{equation*}
    \subKernel f(x) := \E_x [f(X_\muMinEvent)\, \mathbb I_{\muMinEvent<T}].
\end{equation*}
We can then define $\probRegenBeforeMuMin: \cX \to [0,1]$ by
\begin{equation*}
    \probRegenBeforeMuMin(x) := 1- \subKernel 1(x) = 1- \P_x(\muMinEvent<T) = \P_x(T\le \muMinEvent).
\end{equation*}
We can also define a proper Markov kernel: for a measurable set $B\subset \mathcal X$,
\begin{equation*}
    \properKernel(x, B) := \frac{\subKernel(x,B)}{\subKernel(x,\cX)} = \E_x[\mathbb I_B(X_\muMinEvent)|\muMinEvent <T].
\end{equation*}

We need the following technical result; namely, we check \cite[Hypotheses H1, H2]{Benaim2018}.
\begin{lemma}
    $\subKernel$ is Feller, and defining the augmented kernel $\augKern$ on $\mathcal X \cup \{\partial\}$, where $\partial \notin\mathcal X$ represents an absorbing state,
    \begin{equation*}
        \augKern(x,\dif y) := \subKernel(x, \dif y) \mathbb I_{x \in \mathcal X} + \probRegenBeforeMuMin(x) \delta_\partial (\dif y) \mathbb I_{x \in \mathcal X} + \delta_\partial (\dif y) \mathbb I_{x \in \{\partial\}},
    \end{equation*}
    we have that $\partial$ is accessible for $\augKern$.
    \label{lem:H1H2}
\end{lemma}
\begin{proof}
    The Feller property holds by the representation \eqref{eq:subdensity}. Accessibility of $\partial$ is immediate, since started from anywhere, the diffusion path can (eventually) be killed.
\end{proof}

\begin{proposition}\label{prop:trans}
    The Markov chain $Y$ has transition kernel $\mcKern(x,\dif y)$ on $\cX$ given by
    \begin{equation*}
        \begin{split}
            \mcKern(x,\dif y) &= \subKernel(x,\dif y) + \probRegenBeforeMuMin(x) \frac{\mu \subKernel(\dif y)}{\mu \subKernel 1}, \\
            &= \left(1 - \probRegenBeforeMuMin(x) \right) \properKernel(x,\dif y) + \probRegenBeforeMuMin(x) \frac{\mu \subKernel(\dif y)}{\P_\mu(\muMinEvent<T)} .
        \end{split}
    \end{equation*}
\end{proposition}
\begin{proof}
    We have already established in Lemma \ref{lemma:Y_markov_chain} that $Y$ is indeed a Markov chain, and we denote its transition kernel by $\mcKern(x,\dif y)$. Its transition kernel satisfies the following relation: (noting that by continuity, $\P_x(\muMinEvent=T)=0$)
    \begin{equation}
        \begin{split}
        \mcKern(x,\dif y) = \properKernel(x,\dif y) \cdot \P_x(\muMinEvent<T)+ \mu \mcKern(\dif y) \cdot \P_x(\muMinEvent > T).
        \end{split}
        \label{eq:K_mu}
    \end{equation}
    This is because by the memoryless property of the exponential,
    \begin{equation*}
        \mathrm{Law}\left (\xi - \int_0^T \kappa^-(X_u)\dif y\bigg |\xi>\int_0^T \kappa^-(X_u)\dif u\right)= \mathrm{Exp}(1),
    \end{equation*}
    and by the strong Markov property, given $(\muMinEvent > T)$, the subsequent evolution of $X$ at time $T$ is equal in law to $(X|X_0 \sim \mu)$. 
    
    By recursion of \eqref{eq:K_mu}, we arrive at the desired conclusion.
\end{proof}

%%%%%%%%%%%%%%%%%%%%%%%%%%%%%
% Adaptive reinforced process
%%%%%%%%%%%%%%%%%%%%%%%%%%%%%

\subsection{Adaptive reinforced process}
We are interested in studying the limiting behaviour of
\begin{equation*}
    \mu_n := \frac{1}{n+1}\sum_{i=0}^n \delta_{Y_i},
\end{equation*}
where now the $(Y_i)$ are generated as follows:
\begin{equation}
    Y_{n+1}|(Y_0,Y_1,\dots,Y_n) \sim \mcKernDscrtMeas(Y_n,\cdot).
    \label{eq:Yn_reinforcing}
\end{equation}
This corresponds to our adaptive Restore algorithm (Algorithm~\ref{algo:adaptive_restore}), where we are learning an approximation to the minimal rebirth distribution. \correction{We note that here we have not included the $\mu_0$ term which appears in \eqref{eq:mu_t}, but our asymptotic convergence result will also imply convergence when such a term is included, which here would carry a weight of $a/(a+n)$.}

Our goal will be to show the almost sure weak convergence of $\mu_n \to \muMin$ as $n\to \infty$. We will utilise the approach of \cite{Benaim2018}.

In order to understand the limiting properties of the self-reinforcing process $(Y_i)$ in \eqref{eq:Yn_reinforcing}, we need to study the properties of the $\mcKern$ kernels. We write $\mathcal P(\cX)$ for the space of probability measures on $\cX$. We have the following useful representation.
\begin{lemma}
    For any bounded continuous $f$, we have for $x\in \cX$,
    \begin{equation*}
        \mathcal Af(x):=\sum_{n\ge 1} \subKernel^n (x,f) = \int_0^\infty \dif t \int \dif y\, p^{\kappa^+} (t,x,y) f(y)\kappa^-(y).
    \end{equation*}
    The nonnegative kernel $\mathcal A$ is also a bounded kernel on $\cX$: $\mathcal A 1(x)\le \|\mathcal A\|_\infty <\infty$. Furthermore, there exists $\delta>0$ such that for any $x\in \cX$, $\delta \le \mathcal A1(x)$. This implies Lipschitz continuity (with respect to total variation) of the map $\mathcal P(\cX) \to \mathcal P(\cX)$ given by
    \begin{equation*}
        \mu \mapsto \Pi_\mu := \frac{\mu \mathcal A}{\mu \mathcal A 1}.
    \end{equation*}
    \label{lem:repn_Kn}
\end{lemma}

\begin{proof}
    Since we are imposing in Assumption \ref{assm:regularity} that $\kappa^- \le \nK$ uniformly, the law of the sequence of arrivals $\muMinEvent_1, \muMinEvent_2, \dots$ is absolutely continuous with respect to the law of a homogeneous Poisson process $\muMinDomEvent_1, \muMinDomEvent_2, \dots$ of rate $\nK$, which is independent of $X$ and the regenerations. 
    
    Now, we consider
    \begin{equation*}
        \subKernel_{\nK}(x,f):= \E_x[f(X_{\muMinDomEvent_1}) \mathbb I_{\muMinDomEvent_1 < T_1}]= \int_0^\infty \dif t \,\nK\mathrm{e}^{-\nK t}\int \dif y \, p^{\kappa^+}(t,x,y) f(y).
    \end{equation*}
    Here we have made use of the subdensity (14).
    
    Now, using the fact that $ \muMinDomEvent_n \stackrel{\mathrm d}{=} \text{Gamma}(n,\nK)$, or by direct integration, we obtain
    \begin{equation*}
        \subKernel^n_{\nK}(x,f) = \int_0^\infty \dif t \,\frac{(\nK)^n t^{n-1}\mathrm{e}^{-\nK t}}{(n-1)!}\int \dif y\, p^{\kappa^+}(t,x,y) f(y). 
    \end{equation*}
    Therefore,
    \begin{equation*}
        \begin{split}
            \sum_{n\ge 1} \subKernel^n_{\nK} (x,f) &= \int_0^\infty \dif t\, \sum_{n \ge 1}\frac{(\nK)^n t^{n-1}\mathrm{e}^{-\nK t}}{(n-1)!} \int \dif y \, p^{\kappa^+}(t,x,y) f(y), \\
            &= \nK \int_0^\infty \dif t\int \dif y\, p^{\kappa^+}(t,x,y) f(y).
        \end{split}
    \end{equation*}
    We note that this is a finite measure by Lemma \ref{lem:ET1} and Assumption \ref{assm:finiteT1}; we have in fact that $\|\mathcal A\|_\infty \le \nK \sup_{x\in \cX} \E_{x}[T_1]<\infty $.
    
    Now consider $\sum_{n \ge 1} \subKernel^n (x,f)$. By Poisson thinning, we have the representation
    \begin{equation*}
        \begin{split}
            \sum_{n\ge 1} \subKernel^n (x,f) &= \sum_{n\ge 1}Q_{\nK}^n\left (x, f\cdot \frac{\kappa^-}{\nK}\right ), \\
            &= \int_0^\infty \dif t\int \dif y\, p^{\kappa^+}(t,x,y) f(y)\kappa^- (y).
        \end{split}
    \end{equation*} 
    
    The final point follows straightforwardly from compactness of $\cX$ and continuity and positivity of $p^{\kappa^+}(t,x,y)$, and Lipschitz continuity follows similarly to the proof of \citeauthor{Benaim2018} (2018, Proposition 4.5): $\|\Pi_\nu-\Pi_\mu \|_\TV\le 2\|\mathcal A\|/\delta$.
\end{proof}

\begin{proposition}
    Given a fixed probability measure $\mu$ on $\cX$, the invariant distribution of the kernel $\mcKern$ is proportional to $\mu \mathcal A$, where $\mathcal A$ is the kernel $\sum_{n\ge 1}\subKernel^n$.
    \label{prop:propor}
\end{proposition}

\begin{proof}
    First, we have seen from Lemma \ref{lem:repn_Kn} that $\mu \sum_{n\ge 1} \subKernel^n$ is a finite measure. We have the following direct calculation:
    \begin{equation*}
        \begin{split}
            \mu \sum_{n\ge 1} \subKernel^n \mcKern &= \mu \sum_{n\ge 1} \subKernel^n \left (\subKernel + \probRegenBeforeMuMin \frac{\mu \subKernel}{\P_\mu(\muMinEvent<T)}\right ), \\
            &=\mu \sum_{n\ge 1} \subKernel^{n+1} + \left(\mu \sum_{n\ge 1} \subKernel^n \probRegenBeforeMuMin \right) \P_\mu^{-1}(\muMinEvent <T) \mu \subKernel.
        \end{split}
    \end{equation*}
    Since $\probRegenBeforeMuMin = 1 - \subKernel 1$, it follows that
    \begin{equation*}
        \mu \sum_{n \ge 1} \subKernel^n \probRegenBeforeMuMin = \mu \subKernel 1 = \P_\mu(\muMinEvent <T),
    \end{equation*}
    and hence, as desired:
    \begin{equation*}
        \mu \sum_{n\ge 1} \subKernel^n \mcKern = \mu \sum_{n \ge 1} \subKernel^{n+1} + \mu \subKernel = \mu \sum_{n \ge 1} \subKernel^n.
    \end{equation*}
\end{proof}

%%%%%%%%%%%%%%%%%%%%%%%%%%%%%%%%%%%%%%%%%%%%
% Limiting ODE flow and fixed point analysis
%%%%%%%%%%%%%%%%%%%%%%%%%%%%%%%%%%%%%%%%%%%%

\subsection{Limiting ODE flow and fixed point analysis}
\label{sec:ODE}
The limiting flow can be defined just as in \citeauthor{Benaim2018} (2018, Section 5), since we have the appropriate assumptions in force; Lemma~\ref{lem:H1H2} and also the Lipschitz property of $\mu \mapsto \Pi_\mu$, Lemma~\ref{lem:repn_Kn}. In other words, we also have, from \citeauthor{Benaim2018} (2018, Proposition 5.1), an injective semi-flow $\Phi$ on $\mathcal P(\cX)$ such that $t\mapsto \Phi_t(\mu)$ is the unique weak solution to
\begin{equation*}
    \dot \mu_t = -\mu_t +\Pi_{\mu_t},\quad \mu_0 = \mu.
\end{equation*}
We need to check global asymptotic stability, and to do this we will follow the approach in \cite{Wang2020}. In particular, we need to identify the eigenfunctions of $\mathcal A$.

\begin{proposition}\label{prop:minrebirth}
    We have that, for $\muMin$ the minimal rebirth distribution, $\muMin \A \propto \muMin$, and defining $\varphi := \pi|_{\cX}$ to be the restriction of $\pi$ to $\cX$, $\A \varphi = \beta \varphi$, where $\beta := \CMinimum \E_{\muMin}[T_1]$.
\end{proposition}

\begin{proof}
    We directly calculate,
    \begin{equation*}
        \begin{split}
            \muMin \A( y) = \int \muMin (\dif x) \int \dif t \,p^{\kappa^+}(t,x,y) \kappa^-(y)\propto \pi(y)\kappa^-(y) ,
        \end{split}
    \end{equation*}
    since $\int \muMin (\dif x) \int \dif t \,p^{\kappa^+}(t,x,y)\propto \pi(y)$, because the invariant distribution of Restore($L, \kappa^+,\muMin)$ is $\pi$; see equation \eqref{eq:pi_mu*_rep}. Now for the right eigenfunction, we use Tonelli's theorem and symmetry of $p^{\kappa^+}(t,x,y)$ with respect to $x,y$:
    \begin{equation*}
        \begin{split}
            \int_0^\infty \dif t\int_\cX \dif y\, p^{\kappa^+}(t,x,y) \kappa^-(y)\pi(y)&= \int \dif y \, \pi(y) \kappa^-(y)\int_0^\infty \dif t\, p^{\kappa^+}(t,x,y), \\
            &= \int \dif y\, \muMin (y) \CMinimum \int_0^\infty  p^{\kappa^+}(t,y,x) \dif t, \\
            &= \CMinimum \pi(x)\E_{\muMin}[T_1].
        \end{split}
    \end{equation*}
\end{proof}

Given the preceding results, we can now conclude the following.

\begin{proposition}
    We have that $\muMin$ is a global attractor for the semi-flow $\Phi$: we have convergence $\Phi_t(\mu) \to \muMin$ uniformly in $\mu$ in total variation distance.
    \label{prop:global_attractor}
\end{proposition}
\begin{proof}
     Since we have obtained uniform upper and lower bounds $0 < \delta \le \A 1(x)\le \nK \sup_{x\in \cX}\mathbb E_x[T_1]$ from Lemma~\ref{lem:repn_Kn}, the proof is identical to the proof of \citeauthor{Wang2020} (2020, Theorem 3.6), and hence for brevity here we only provide a sketch. The semiflow $\Phi_t(\mu)$ can be realised as a time-change of the operator $\tilde \Phi_t(\mu):= \mu \exp(t\mathcal A)/\mu \exp(t\mathcal A)1$. Uniform convergence of $\tilde \Phi_t(\mu)$ to $\mu^+$ can then be shown directly and straightforwardly using a drift condition, since the state space is compact, which carries over to $\Phi_t(\mu)$ because of our aforementioned uniform upper and lower bounds.
\end{proof}

%%%%%%%%%%%%%%%%%%%%%%%%%%%%%%%%
% Asymptotic pseudo-trajectories
%%%%%%%%%%%%%%%%%%%%%%%%%%%%%%%%

\subsection{Asymptotic pseudo-trajectories}
We secondly need to demonstrate that our trajectories $(n\mapsto \mu_n)$, once suitably embedded in continuous time, are an asymptotic pseudo-trajectory for the semi-flow $\Phi$ defined in Section~\ref{sec:ODE}.

The key technical challenge to establishing this is to prove an analogue of \citeauthor{Benaim2018} (2018, Lemma 6.2) in our setting. Once that is in place, everything else follow identically from \citeauthor{Benaim2018} (2018, Section 6).

We need the following Lipschitz property, where the total variation norm for a signed measure $\nu$ on $\mathcal X$ is defined as
\begin{equation*}
    \|\nu\|_\TV := \sup \{|\nu(f)|: f: \mathcal X \to \R \text{ bounded measurable, }\|f\|_\infty \le 1 \}.
\end{equation*}
\begin{lemma}
\label{lemma:lip}
For some $C_\mathrm L>0$, for probability measures $\mu, \nu \in \mathcal P(\cX)$ and $j\in \mathbb N$,
\begin{equation*}
        \sup_{\alpha \in \mathcal P(\cX)}\| \alpha \mcKern^j -\alpha \subKernel_\nu^j \|_{\TV}\le C_\mathrm L 2^j \|\mu-\nu\|_\TV, 
\end{equation*}
    and for each bounded function $f$, 
    \begin{equation*}
        \sup_{x \in\cX}\| \mcKern^j(x,f) - \subKernel_\nu^j(x,f) \|_{\TV}\le C_\mathrm L 2^j\|f\|_\infty \|\mu-\nu\|_\TV.
    \end{equation*}
\end{lemma}
\begin{proof}
    Recall that
    \begin{equation*}
        \mcKern(x,\dif y) = \subKernel(x,\dif y) + \probRegenBeforeMuMin(x) \frac{\mu \subKernel(\dif y)}{\mu \subKernel 1}.
    \end{equation*}
    So fix $\alpha, \mu, \nu\in \mathcal P(\cX)$. Then we have
    \begin{equation*}
        \begin{split}
            \| \alpha \mcKern -\alpha \subKernel_\nu \|_{\TV} &= \left\|\alpha(\probRegenBeforeMuMin) \frac{\mu \subKernel}{\mu \subKernel 1} - \alpha(\probRegenBeforeMuMin) \frac{\nu \subKernel}{\nu \subKernel 1}  \right\|_\TV, \\
            &= \alpha(\probRegenBeforeMuMin) \left \|\frac{\mu \subKernel}{\mu \subKernel 1} -  \frac{\nu \subKernel}{\nu \subKernel 1} \right \|_\TV, \\
            &\le \left \|\frac{\mu \subKernel}{\mu \subKernel 1} -  \frac{\nu \subKernel}{\nu \subKernel 1}\right \|_\TV,
        \end{split}
    \end{equation*}
    since $\alpha(\probRegenBeforeMuMin)\le 1$. So we need to bound this final term. Now, since $Q$ is Feller (Lemma \ref{lem:H1H2}), and $Q1>0$ pointwise on $\mathcal X$, compactness we can find $\delta>0$ such that $\nu Q 1\ge \delta$ for any $\nu\in \mathcal P(\mathcal X)$. Thus:
    \begin{equation}
        \begin{split}
            \left \|\frac{\mu \subKernel}{\mu \subKernel 1} -  \frac{\nu \subKernel}{\nu \subKernel 1}\right \|_\TV &= \left \| \frac{\mu \subKernel (\nu \subKernel 1) - (\mu \subKernel 1)\nu \subKernel}{\mu \subKernel 1 \, \nu \subKernel 1} \right \|_\TV, \\
            &\le \frac{\|\mu \subKernel\|_\TV}{\mu \subKernel 1\, \nu \subKernel 1}|\nu \subKernel 1 - \mu \subKernel 1|+\frac{|\mu \subKernel 1| \|\mu \subKernel - \nu \subKernel\|_\TV}{\mu \subKernel 1 \, \nu \subKernel 1}, \\
            &\le \frac{\|\mu - \nu\|_\TV \|\subKernel 1\|_\infty }{\delta } + \frac{\|\mu- \nu\|_\TV }{\delta}, \\
            &\le \frac{2}{\delta} \|\mu-\nu\|_\TV.
        \end{split}
        \label{eq:lem4.9.0}
    \end{equation}

    Now, we with to establish the following statement: for each $j\in \mathbb N$,
    \begin{equation}
        \sup_{\alpha \in \mathcal P(\cX)}\| \alpha \mcKern^j -\alpha \subKernel_\nu^j \|_{\TV}\le \frac{2}{\delta}\cdot 2^j \|\mu-\nu\|_\TV. 
        \label{eq:lemma4.9}
    \end{equation}
    This immediately implies the second statement of the Lemma, by definition of total variation.

    Given \eqref{eq:lem4.9.0}, the statement \eqref{eq:lemma4.9} follows directly from an inductive argument on $j\in \mathbb N$, this calculation is fully detailed in \citeauthor{Benaim2018} (2018, [Proof of Lemma 6.2]).
\end{proof}

With this result, the rest of the approach of \citeauthor{Benaim2018} (2018, Section 6) goes through to establish the desired result.

\begin{theorem}
    Under our Assumptions~\ref{assm:regularity}--\ref{assm:finiteT1}, almost surely, the sequence $(\mu_n)$ converges weakly to the minimal regeneration distribution $\muMin$.
\end{theorem}
\begin{proof}
    By embedding the sequence $(\mu_n)$ into continuous time as in \citeauthor{Benaim2018} (2018, Section 6.1), the resulting process $(\hat \mu_t)$ is an asymptotic pseudo-trajectory of $\Phi$, by Lemma~\ref{lemma:lip} and \citeauthor{Benaim2018} (2018, Theorem 6.4). Combined with Proposition~\ref{prop:global_attractor}, this proves the result; see \cite{Benaim1999}.
\end{proof}

%%%%%%%%%%
% Examples
%%%%%%%%%%

\section{Examples} \label{sec:examples}

The examples presented below show that even when it is possible to simulate a standard Restore process (because a suitably large constant $C$ is known in advance), adaptive Restore can significantly decrease the MSE of some expectations of interest. However, they also signal that convergence of the process can be slow for high-dimensional or multi-modal target distributions.

%%%%%%%%%%%%%%%%%%%%%%%%%%%%%%%%%%%%%%%%%%%%
% Logistic Regression Model of Breast Cancer
%%%%%%%%%%%%%%%%%%%%%%%%%%%%%%%%%%%%%%%%%%%%

\subsection{Logistic regression model of breast cancer} \label{sec:log_reg_breast_cancer}

We experimented with using Restore to simulate from the (transformed) posterior of a Logistic Regression model of breast cancer ($d=10$). For this model, we denote the data, consisting of predictor-response pairs, by $\{ (z_i, y_i) \}_{i=1}^n$. The random variables of interest are the regression coefficients $\beta = (\beta_1, \beta_2, \dots, \beta_d)^T$. Details of the data and prior are given in Supplement C \citep{mckimm2024sampling}. The likelihood of a Logistic Regression model is: 
\begin{equation*}
    l(\{ (z_i, y_i) \}_{i=1}^n | \beta) = \left[ \prod_{i=1}^n \frac{1}{1 + \exp\{-y_i \beta^T z_i\} } \right].
\end{equation*}

Four types of Restore process were experimented with: standard Restore with $\mu \equiv \N(0, I)$, standard Restore with $\mu \equiv \N(0, 1.5I)$, adaptive Restore with $\muMin \equiv \N(0, I)$ and adaptive Restore with $\muMin \equiv \muMin_{\N}$, where $\muMin_{\N}$ is the minimal regeneration distribution of an isotropic Gaussian distribution. We investigated how the scale of the covariance matrix of $\mu$ affects $\regenRateTrunc$, by estimating how $\regenRateTrunc$ varies when $\mu \equiv \N(0, \alpha I)$, with different values of $\alpha$. We found that as a function of $\alpha$, $\regenRateTrunc$ is close to its minimum at $\alpha = 1.5$, as shown by Figure \ref{fig:lr_kappa_scale_regen_dist}. Hence we experimented with using $\mu \equiv \N(0, 1.5 I)$ for standard Restore.

\begin{figure}
    \centering
    \begin{subfigure}{0.5\textwidth}
        \centering
        \includegraphics[width=\linewidth, trim={0 0.5cm 0 1.5cm}, clip]{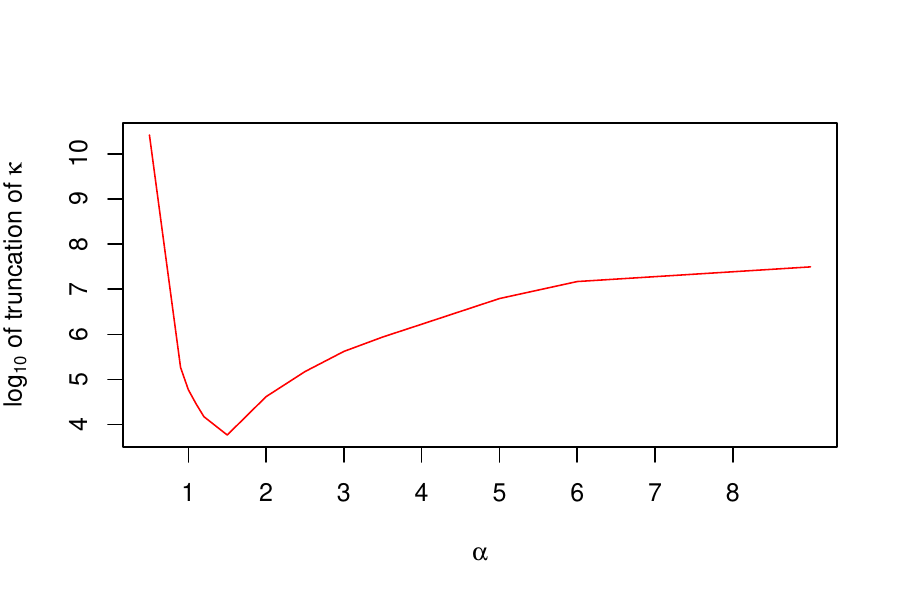}
        \caption{Logistic regression model.}
        \label{fig:lr_kappa_scale_regen_dist}
    \end{subfigure}%
    \begin{subfigure}{0.5\textwidth}
        \centering
        \includegraphics[width=\linewidth, trim={0 0.5cm 0 1.5cm}, clip]{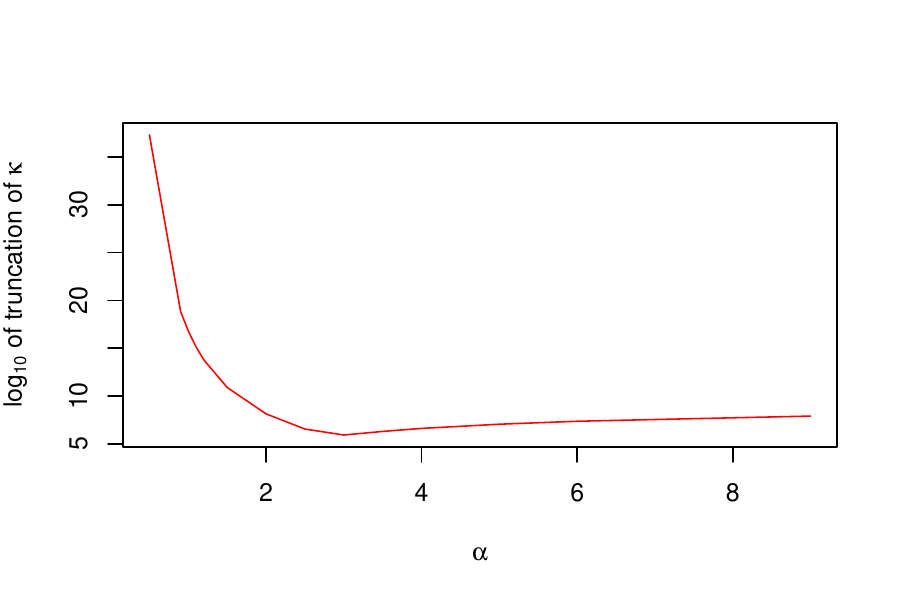}
        \caption{Pump failure model.}
        \label{fig:pump_kappa_scale_regen_dist}
    \end{subfigure}
    \caption{Plot of $\log_{10}\regenRateTrunc$ as a function of $\alpha$ when $\mu \equiv \N(0, \alpha I)$ and $\pi$ is the posterior of a logistic regression model of breast cancer (\ref{fig:lr_kappa_scale_regen_dist}) or a hierarchical model of pump failure (\ref{fig:pump_kappa_scale_regen_dist}). Both posteriors are pre-transformed based on their Laplace approximations. $\regenRateTrunc$ satisfies $\P(\kappa(X) < \regenRateTrunc) \approx 0.9999$.}
\end{figure}

For each type of Restore process, we simulated 100 sample paths. The Monte Carlo samples produced by these paths were then used to estimate $\E[X_i]$ and $\E[X_i^2]$ for $i=1,\dots,10$. These estimates were then used to estimate the MSE of estimates of $\E[X_i]$ and $\E[X_i^2]$ for $i=1,\dots,10$. To estimate the MSE, we took as `ground truth' the estimates generated by a Markov chain of length $5 \times 10^8$, thinned so that $10^7$ samples were used in the Monte Carlo sum for the estimate.

For both $\mu_0 \equiv \N(0, I)$ and $\mu_0 \equiv \muMin_{\N}$ the adaptive Restore processes were simulated with $\nK=5.12$ (3.s.f), $\minRegenRateTrunc=31.7$ (3.s.f), $a=10, \outputRate=1.0$ and with Short-Term Memory with $\nCloud=10^4$ and $\nForget=10$. Parameter $\minRegenRateTrunc$ was chosen so that $\P(\pkap < \minRegenRateTrunc) \approx 0.9999$. The processes were simulated for time $T_{\text{adapt}} = 4 \times 10^5$, of which $b=3 \times 10^5$ was burnt.

The times $T_{\text{adapt}}$ and $b$ were chosen by first setting $T=10^6$ and making initial runs with $\outputRate = 0.1$ (this small output rate was used, so that the program implementing the sampler required less memory, which was important for simulating the sample paths on a server). Each process took roughly 100 minutes to simulate. Visually inspecting boxplots of the estimates of $\E[X_i]$ and $\E[X_i^2]$ for $i=1,\dots,10$, shown in Supplement C \citep{mckimm2024sampling}, the process appears to have approximately converged by time $3 \times 10^5$.

For the standard Restore processes with $\mu \equiv \N(0, I_d)$ and $\mu \equiv \N(0, 1.5 I)$, $\regenRateTrunc$ was chosen to satisfy $\P(\kappa(X) < \regenRateTrunc) \approx 0.9999$. The total simulation time, $T_{\text{std}}$, was chosen so that $T_{\text{std}} \times \regenRateTrunc = T_{\text{adapt}} \times (\minRegenRateTrunc + \nK)$, and thus in expectation the adaptive and standard processes evaluate their respective rates the same number of times. The output rate was chosen so that, in expectation, $10^5$ output states were generated by each process (the same as the number of output states produced by the corresponding adaptive processes, post burn-in.)

Supplement C \citep{mckimm2024sampling} contains a comparison of the MSE of estimates of $\E[X_i]$ and $E[X_i^2]$, $(i=1,\dots,10)$ for the four types of Restore processes simulated. Table \ref{tab:lr_mse} gives the values of the mean MSE of $\E[X_i]$ and $\E[X_i^2]$. Using standard Restore with $\mu \equiv \N(0, 1.5 I)$ produces good estimates of the first moment, but relatively poor estimates of the second moments. By contrast, adaptive Restore with $\mu_0 \equiv \N(0,I)$ gives good estimates of both moments. 

\begin{table}
    \centering
    \begin{tabular}{| c | c | c |} 
        Type & Mean MSE $\E[X_i]$ & Mean MSE $\E \big[ X_i^2 \big]$ \\
        $\mu_0 \equiv \N(0,I)$ & $2.0 \times 10^{-4}$ & $2.3 \times 10^{-4}$ \\
        $\mu_0 \equiv \muMin_{\N}$ & $2.1 \times 10^{-4}$ & $2.4 \times 10^{-4}$ \\
        $\mu \equiv \N(0,I)$ & $2.9 \times 10^{-4}$ & $2.4 \times 10^{-3}$\\
        $\mu \equiv \N(0, 1.5 I)$ & $1.4 \times 10^{-4}$ & $6.3 \times 10^{-4}$
    \end{tabular}
    \caption{Mean MSE of the first and second moments of the $d=10$ variables, computed using different versions of Restore. Types $\mu$ and $\mu_0$ indicate standard and adaptive Restore processes.}
    \label{tab:lr_mse}
\end{table}

%%%%%%%%%%%%%%%%%%%%%%%%%%%%%%%%%%%%
% Hierarchical Model of Pump Failure
%%%%%%%%%%%%%%%%%%%%%%%%%%%%%%%%%%%%

\subsection{Hierarchical model of pump failures} \label{sec:pump_model}

For the hierarichal model of pump failure, introduced in Section \ref{sec:adaptive_restore}, three types of Restore processes were experimented with: standard Restore with $\mu \equiv \N(0, 3I)$, adaptive Restore with $\muMin \equiv \N(0, I)$ and adaptive Restore with $\muMin \equiv \muMin_{\N}$. We did not experiment with using standard Restore with $\mu \equiv \N(0, I)$ because, as explained in \ref{sec:adaptive_restore}, due to the mismatch of this $\mu$ with $\pi$, the regeneration rate is extremely large. We investigated how $\regenRateTrunc$ must vary to satisfy $\P(\kappa(X) < \regenRateTrunc) \approx 0.9999$ for $\mu \equiv \N(0, \alpha I)$, with different values of $\alpha$. We found that as a function of $\alpha$, $\regenRateTrunc$ is close to its minimum at $\alpha = 3$, as shown by Figure \ref{fig:pump_kappa_scale_regen_dist}. Hence we experimented with using $\mu \equiv \N(0, 3 I)$ for standard Restore.

As for the example of a logistic regression model of breast cancer, for each type of Restore process we simulated 100 samples paths. Each sample path was used to estimate $\E[X_i]$ and $\E[X_i^2]$ for $i=1,\dots,11$. Then these estimates were used to estimate the MSE of these quantities, taking as `ground truth' an estimate computed using the samples generated by a Markov chain, with near zero autocorrelation, of length $10^7$.

For both $\mu_0 \equiv \N(0,I)$ and $\mu_0 \equiv \muMin$ the adaptive Restore processes were simulated with $\minRegenRateTrunc=25.0, \nK=5.43$ (3.s.f), $a=10$ and $\outputRate=1$. The processes had Short-Term Memory with $\nCloud=10^4$ and $\nForget=10$. The total simulation time was $T_{\text{adapt}} = 3 \times 10^5$, of which $b = 2 \times 10^5$ was burn-in time. To decide on an appropriate burn-in and simulation times, we initially simulated process with $T_{\text{adapt}}=10^6$ and $\outputRate=0.1$. Dividing the samples into 10 batches and inspecting these visually, see Supplement F \citep{mckimm2024sampling}, we estimated that $b =2 \times 10^5$ was an appropriate burn-in time.

For the standard Restore processes with $\mu \equiv \N(0, 3 I), \regenRateTrunc$ was chosen to satisfy $\P(\kappa(X) < \regenRateTrunc) \approx 0.9999$. The total simulation time, $T_{\text{std}}$, was chosen so that $T_{\text{std}} \times \regenRateTrunc = T_{\text{adapt}} \times (\minRegenRateTrunc + \nK)$, and thus in expectation the adaptive and standard processes evaluate their respective rates the same number of times. The output rate was chosen so that, in expectation, $10^5$ output states were generated by each process (the same as the number of output states produced by the corresponding adaptive processes, post burn-in.)

Supplement F \citep{mckimm2024sampling} contains a comparison of the MSE of estimates of $\E[X_i]$ and $\E[X_i^2]$, ($i=1,\dots,11$) for the three Restore processes simulated. Table \ref{tab:pump_mse} gives the values of the mean MSE of $\E[X_i]$ and $\E[X_i^2]$. Adaptive Restore clearly out-performs standard Restore.

\begin{table}
    \centering
    \begin{tabular}{| c | c | c |} 
        Type & Mean MSE $\E[X_i]$ & Mean MSE $\E \big[ X_i^2 \big]$ \\
        $\mu_0 \equiv \N(0,I)$ & $2.6 \times 10^{-4}$ & $3.5 \times 10^{-4}$ \\
        $\mu_0 \equiv \muMin_{\N}$ & $2.6 \times 10^{-4}$ & $3.4 \times 10^{-4}$ \\
        $\mu \equiv \N(0, 3 I)$ & $1.6 \times 10^{-3}$ & $1.3 \times 10^{-2}$
    \end{tabular}
    \caption{Mean MSE of the first and second moments of the $d=11$ variables, computed using different versions of Restore. Types $\mu$ and $\mu_0$ indicate standard and adaptive Restore processes.}
    \label{tab:pump_mse}
\end{table}

%%%%%%%%%%%%%%%%%%%%%%%%%%%%%%%%%%%%%%
% Log-Gaussian Cox Point Process Model
%%%%%%%%%%%%%%%%%%%%%%%%%%%%%%%%%%%%%%

\subsection{Log-Gaussian Cox point process model} \label{sec:log_gauss_cox_model}

A Log-Gaussian Cox Point Process models a $[0,1]^2$ area divided into a $n \times n$ grid. The number of points $Y = \{Y_{i,j}\}$ in each cell is conditionally independent given latent intensity $\Lambda = \{\Lambda_{i,j}\}$ and has Poisson distribution $n^2 \Lambda_{i,j} = n^2 \exp \{ X_{i,j} \}$. The latent field is $X = \{X_{i,j}\}$. Assume $X$ is a Gaussian process with mean vector zero and covariance function $\Sigma_{(i,j),(i',j')} = \exp \{ -\delta(i,i';j,j')/n\}$, where $\delta(i,i';j,j') = ((i-i')^2 + (j-j')^2)^{1/2}$. We have:
\begin{equation*}
    \log \pi(x | y) = \sum_{i,j} y_{i,j} x_{i,j} - n^2 \exp \{x_{i,j} \} - \frac{1}{2} x^T \Sigma^{-1} x + \text{const}.
\end{equation*}

We present results for simulated data on a 5 by 5 grid, so $d=25$. After transformation, see Supplement B \citep{mckimm2024sampling}, the posterior distribution of this model is close to an Isotropic Gaussian distribution. For standard Restore setting $\mu \equiv \mathcal{N}(0, I)$ results in $\P[\kappa(X) < \correction{67.1}] \approx 0.9999$. Thus $\regenRateTrunc = \correction{67.1}$ would be appropriate.

For adaptive Restore we have $\P[\pkap < 18.5] \approx 0.9999$. Thus adaptive Restore reduces the necessary truncation level \correction{of the} regeneration rate by a factor of \correction{less than 4}. Simulation runs indicate that convergence of the adaptive process for this $d=25$ posterior is slow.
\correction{In particular, we found that the process had still not converged after a simulation time of $T=10^7$, which took about an hour of wall-clock time.} Though $\mu_t$ does not need to adapt to account for skew so much, it still needs to change significantly so that it is centred correctly -- this is harder in higher dimensions.

%%%%%%%%%%%%%%%%%%%%%%%%%%%%%%%%%%%
% Mixture of Gaussian distributions
%%%%%%%%%%%%%%%%%%%%%%%%%%%%%%%%%%%

\subsection{Mixture of Gaussian distributions} \label{sec:gaussian_mixture}

Multi-modal posterior distributions sometimes arise in Bayesian modelling problems. For example, the standard two-parameter Ising model \citep{geyer1991markov} is bimodal for some parameter combinations; a model of a problem concerning sensor network location \citep{ihler2005nonparametric} is a popular example that features in many papers \citep{ahn2013distributed, lan2014wormhole, pompe2020adaptive}.
%astrophysics \citep{jones2014disentangling}
Standard MCMC algorithms struggle to sample multi-modal distributions because the area of low probability density between modes acts as a barrier that is difficult to cross. Several techniques have been developed specifically for multi-modal posteriors, which generally fall under tempering \citep{geyer1991markov, marinari1992simulated} and mode-hopping strategies \citep{tjelmeland2001mode, ahn2013distributed}.

We explore the use of an adaptive Restore process for simulating from the Gaussian mixture distribution $\pi(x) = w_1 \mathcal{N}(x; m_1, \Sigma_1) + w_2 \mathcal{N}(x; m_2, \Sigma_2)$, for $w_1=0.4, w_2=0.6, m_1 = (1.05, 1.05), m_2 = (-1.05, -1.05)$,
\begin{equation*}
    \Sigma_1 =
    \begin{pmatrix}
        1 & -0.1 \\
        -0.1 & 1
    \end{pmatrix} \text{ and }
    \Sigma_2 =
    \begin{pmatrix}
        1 & 0.1 \\
        0.1 & 1
    \end{pmatrix}.
\end{equation*}
Figure \ref{fig:gauss_mix_pi_kappa} shows contour plots of the density of $\pi$ and $\pkap$. In particular, figure \ref{fig:gauss_mix_kappa_partial} shows that the region for which $\pkap$ is zero, which corresponds to the support of $\muMin$, consists of two separate non-connected areas.
For standard Restore and adaptive Restore, we set $\mu$ and $\mu_0$ respectively to $\mathcal{N}(0, 3I)$.

\correction{We simulated 100 adaptive Restore processes with Short-Term Memory and parameters $\nK=1.0, \minRegenRateTrunc = 7.79$ (3.s.f), $\outputRate=1.0, a = 100, \nCloud=10^4, \nForget=2, T_{\text{adapt}}=4 \times 10^5$ and burn-in period $b=3 \times 10^5$. The parameters $T_{\text{adapt}}$ and $b$ were chosen based on preliminary runs; $\pK$ was chosen so that $\P(\pkap(X) < \pK) \approx 0.9999$.}

\correction{As a comparison, 100 standard Restore processes were simulated. The parameter $\regenRateTrunc$ was set as 308 (3.s.f) to satisfy $\P(\kappa(X) < \regenRateTrunc) \approx 0.9999$. As in subsection \ref{sec:log_reg_breast_cancer}, the simulation time $T_{\text{std}}$ was chosen (as $11.4 \times 10^4$) so that $T_{\text{std}} \times \regenRateTrunc = T_{\text{adapt}} \times (\minRegenRateTrunc + \nK)$, and thus in expectation the adaptive and standard processes evaluate their respective rates the same number of times. The output rate was chosen so that each sample path generated $10^5$ output states, in expectation.}

\correction{Although the truncation of the regeneration rate is far smaller for adaptive Restore than for standard Restore, reduced by a factor of approximately 40, standard Restore actually performs better. The MSE for the estimates of the expectation of the first marginal is $1.25 \times 10^{-4}$ for standard Restore and $9.72 \times 10^{-4}$ for adaptive Restore. Thus the MSE for the standard Restore estimates is 87\% lower than the MSE for the adaptive Restore estimates. The fact that standard Restore performs better for this example is in large part due to the slow convergence of adaptive Restore for this multimodal target distribution.}

In analogue to similar schemes based on stochastic approximation \citep{Aldous1988, Blanchet2016, Benaim2018, Mailler2020},
this slow convergence is a result of the urn-like behaviour intrinsic to such methods. Although the chain is guaranteed to converge asymptotically, in finite time the chain is naturally inclined to visit regions it has visited before. For example, even on a \textit{finite} state space, \citeauthor{Benaim2015a} (2015, Corollary 1.3) shows convergence can occur in some cases at a very slow polynomial rate.

Thus in practice, we suggest a judicious choice of initial distribution $\mu_0$ and constant $a$ as in Algorithm~\ref{algo:adaptive_restore}, to ensure that the measures $\mu_t$ quickly place some mass in all modes of the target distribution.

\begin{figure}
    \centering
    \begin{subfigure}{0.475\textwidth}
        \centering
        \includegraphics[width=\textwidth, trim = {0 0.5cm 0 2cm}, clip]{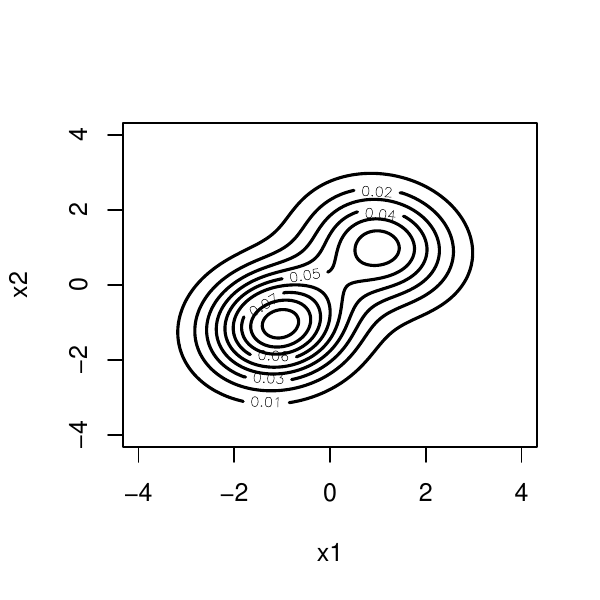}
        \caption{$\pi$}
    \end{subfigure}%
    \begin{subfigure}{0.475\textwidth}
        \centering
        \includegraphics[width=\textwidth, trim = {0 0.5cm 0 2cm}, clip]{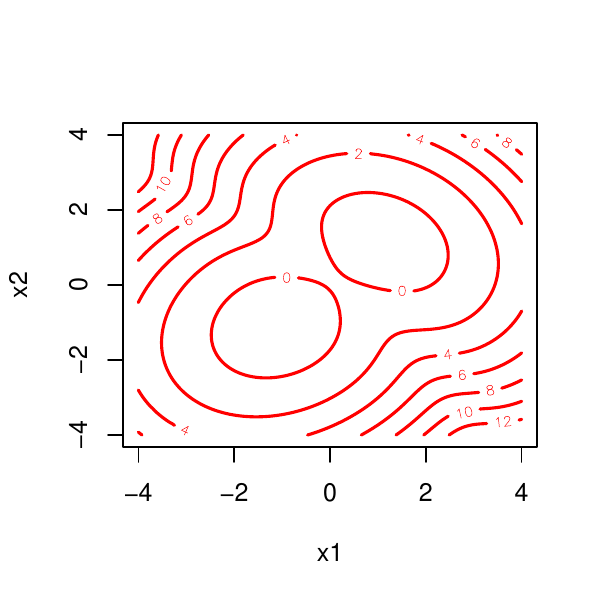}
        \caption{$\pkap$}
        \label{fig:gauss_mix_kappa_partial}
    \end{subfigure}
    \caption{Density of $\pi$ and $\pkap$ for $\pi$ a mixture of bivariate Gaussian distributions.}
    \label{fig:gauss_mix_pi_kappa}
    % 4 by 4 inch figures
\end{figure}

%%%%%%%%%%%%%%%%%%%%%%%%%%%%%
% Multivariate t-distribution
%%%%%%%%%%%%%%%%%%%%%%%%%%%%%

\subsection{Multivariate t-distribution}

Recall that a $d$-dimensional multivariate t-distribution with mean $m$, scale matrix $\Sigma$ and $\nu$ degrees of freedom has density: $\pi(x) \propto \left[ 1 + \nu^{-1} (x - m)^T \Sigma^{-1} (x - m) \right]^{-(\nu + d)/2}$. We consider sampling from a bivariate t-distribution with $\nu = 10$, zero mean and identity scale matrix. We then have $\pkap(x) < 1.55, \forall x \in \mathbb{R}^d$ (this bound is not tight), so can take $\pK = 1.55$ (no need to truncate $\kappa$). In general, Restore processes are particularly well suited to simulating from t-distributions, since the regeneration rate is naturally bounded. For this example, we can take $\nK = 1.2$. The process is quickly able to recover the true variance of each marginal of the target distribution, which is $\nu/(\nu-2)$. Figure \ref{fig:mvt} shows contours of $\tilde{\kappa}, \pkap$ and $\muMin$. A notable feature is that, moving outwards from the origin, $\pkap$ rises to its maximum value then asymptotically tends to zero.

\begin{figure}
    \centering
    \begin{subfigure}{0.45\textwidth}
        \centering
        \includegraphics[width=\textwidth, trim = {0 0.5cm 0 2cm}, clip]{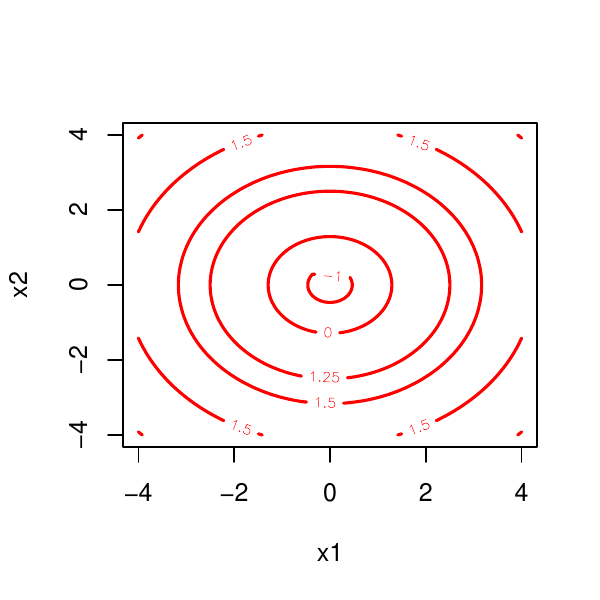}
        \caption{$\tilde{\kappa}$}
    \end{subfigure}%
    \begin{subfigure}{0.45\textwidth}
        \centering
        \includegraphics[width=\textwidth, trim = {0 0.5cm 0 2cm}, clip]{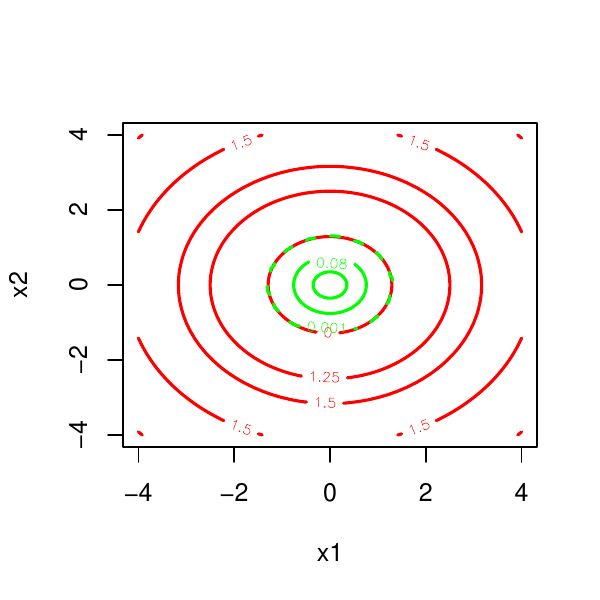}
        \caption{$\pkap$ and $\muMin$}
    \end{subfigure}
    \caption{Contours of $\tilde{\kappa}, \pkap$ and $\muMin$ for $\pi$ a bivariate t-distribution with $\nu = 10$.}
    \label{fig:mvt}
    % 4 by 4 inch figures
\end{figure}

A t-distribution serves as a good test case to check the convergence of implementations of adaptive Restore, because $\pkap$ has a bound, so does not need to be truncated. However, when $\mu$ is Gaussian, $\kappa$ will also have a relatively small upper bound, so in this case adaptive Restore provides no significant advantage over standard Restore.

%%%%%%%%%%%%
% Discussion
%%%%%%%%%%%%

\section{Discussion} \label{sec:discussion}

This article has introduced the adaptive Restore process, an extension of the Restore process \citep{Wang2021}, which adapts the regeneration distribution on the fly. We have used empirical experiments to compare the performance of adaptive and standard Restore for a number of target distributions. From a theoretical perspective, we have demonstrated how the framework of stochastic approximation can be successfully applied to a novel class of Markov processes.

We have shown that the regeneration rate is extremely sensitive to the choice of $\mu$. In fact, when $\pi$ is the posterior of a model of pump failures and $\mu \equiv \N(0, I)$ then $\kappa$ becomes so large that simulation is numerically infeasible. The examples of subsections \ref{sec:log_reg_breast_cancer} and \ref{sec:pump_model}, of a Logistic regression model of breast cancer and a hierarchical model of pump failures, show that adaptive Restore can guard against poor choices of a fixed $\mu$, for which $\kappa$ can be extremely large in parts of the space. While frequent regeneration is not in itself a bad thing, frequent regeneration into regions of low probability mass is computationally wasteful.

If $\mu$ were fixed as a sophisticated approximation of $\pi$, as the dimension $d$ of $\pi$ increases, the quality of this approximation would deteriorate, again leading to $\kappa$ becoming extremely large in parts of the space. Furthermore, it would only be possible to adapt $\mu$ to become closer to $\pi$, based on samples from $\pi$ generated using a Restore process, if the initial choice of $\mu$ were close enough to $\pi$ for simulation to be tractable and a sufficiently large constant $C$ were known. Crucially, as $d$ increases, $\muMin$ remains close to $\pi$. Indeed, it is shown in Wang (\citeyear[Section 5.6]{Wang2020thesis}) for several examples that $\muMin$ has a stable behaviour in the high-dimensional $d\to\infty $ limit. This means that, unlike other methods making use of global regenerative moves via the independence sampler and Nummelin splitting \citep{nummelin1978splitting, mykland1995regeneration}, global moves are more likely to be to areas of the space where $\pi$ has significant mass. However, the example in subsection \ref{sec:log_gauss_cox_model} showed empirically that the adaptive Restore process was slow to converge for the posterior of a Log-Gaussian Cox point process model with dimension $d=25$. Future work may look into ways of improving the convergence of the process for higher dimensional target distributions.

Some properties of standard Restore that are unfortunately not inherited by adaptive Restore are independent and identically distributed tours, an absence of burn-in period and the ability to estimate normalizing constants, unless a fixed regeneration distribution is used in parallel. Moreover, convergence appears to be slow for multi-modal distributions, as demonstrated by the example in subsection \ref{sec:gaussian_mixture}. Since tours begin with distribution $\mu_t$ and this distributed changes over time, tours are no longer independent and identically distributed. A burn-in period is required, during which $\mu_t$ converges to $\muMin$ and $\pi_t$, the stationary distribution of the process at time $t$, converges to $\pi$. For standard Restore, $\tilde{C}$, the regeneration constant with the normalizing constant $Z$ absorbed, is defined explicitly and hence can be used to recover $Z$. On the other hand, for adaptive Restore this constant is defined implicitly and thus cannot be used to recover $Z$, unless adaptivity is stopped for that purpose.

Despite these downsides, adaptive Restore represents a significant improvement on standard Restore by making simulation tractable for a wider range of target distributions. We have shown that simulation of mid-dimensional target distributions is practical with adaptive Restore and have presented a novel application of stochastic approximation to establish convergence of a self-reinforcing process.

%\section*{Supplementary Material}
%
%The supplementary material \citep{mckimm2024sampling} provides: a discussion of the output of Restore samplers, details of how to make a pre-transformation  $\pi$, information on the Logistic Regression model of breast cancer, a formal justification for adapting $\mu_t$ to converge to $\muMin$, and further examples.
%
%\clearpage

\begin{funding}

M. Pollock and G.O. Roberts are partially supported by the EPSRC grant PINCODE (EP/X028119/1) and by the UKRI grant OCEAN (EP/Y014650/1). G.O Roberts is also partially supported by other EPSRC grants: EP/V009478/1 and EP/R018561/1. A. Wang has been partially supported by EPSRC grant EP/V009478/1, and H. McKimm through EPSRC studentship EP/L016710/1 (OxWaSP) and postdoctoral fellowship EP/W015080/1 (Imperial College London). C.P. Robert is partially funded by the European Union (ERC-2022-SYG-OCEAN-101071601) and by a Prairie chair from the Agence Nationale de la Recherche (ANR-19-P3IA-0001).
\end{funding}

%\bibliographystyle{imsart-nameyear.bst}
%\bibliography{adaptive_restore.bib}

\input newarXiv.bbl
\clearpage
% 6/2/24 revision ROUND 3 for Bernoulli
%%%%%%%%%%
% Appendix
%%%%%%%%%%

\begin{appendix}

%%%%%%%%
% Output
%%%%%%%%

\section{Output} \label{apendix:output}

When output times are fixed, let $\{ t_1, t_2, \dots \}$ be an evenly spaced mesh of times, with $t_i = i \Delta$ for $i=1,2,\dots$ and $\Delta > 0$ some constant. When output times are random, let $\{ t_1, t_2, \dots \}$ be the events of a homogeneous Poisson process with rate $\outputRate > 0$. In either case, the output of the process is $\{ X_{t_1}, X_{t_2}, \dots \}$. Suppose there are $n$ output states, then we estimate expectations using the unbiased approximation:
\begin{equation*}
    \pi[f] \approx \frac{1}{n} \sum_{i=1}^n f(X_{t_i}).
\end{equation*}
Algorithmically, there is little difference between using fixed and random output times.  The memoryless property of Poisson processes allows one to generate the next potential regeneration and output events, $\timeToNextPotentialRegen$ and $s$, simulate the process forward in time by $\timeToNextPotentialRegen \vee s$, then discard both $\timeToNextPotentialRegen$ and $s$. When using a fixed mesh of times, the memoryless property no longer applies, so one must keep track of the times of the next output and potential regeneration events. 

%%%%%%%%%%%%%%%%%%%%%%%%%%%%%%%%%%%%%%%%%%%%%%%
% Pre-transformation of the target distribution
%%%%%%%%%%%%%%%%%%%%%%%%%%%%%%%%%%%%%%%%%%%%%%%

\section{Pre-transformation of the target distribution} \label{sec:targ_pre-transform}

In the multi-dimensional setting, the Brownian Motion Restore sampler is far more efficient at sampling target distributions for which the correlation between variables is small. Rate $\tilde{\kappa}$ is more symmetrical for target distributions $\pi$ with near-symmetrical covariance matrices. Since the Markov transition kernel for Brownian motion over a finite period of time is symmetrical, local moves are better suited to near-symmetrical target distributions.

More generally, the parameterization of $\pi$ has a large effect on Bayesian methods \citep{hills1992parameterization}. In practice, we recommend making a transformation so that the transformed target distribution has mean close to zero and covariance matrix close to the identity. Suppose we have $X \sim \mathcal{N}(m, \Sigma)$ and that $\Sigma = V \Lambda V^T$ for $V$ a matrix with columns the eigenvectors of $\Sigma$ and the corresponding eigenvalues forming a diagonal matrix $\Lambda$. Then for $X' \sim \mathcal{N}(0, I_d)$, we have $X = \Sigma^{1/2}X' + m$, where $\Sigma^{1/2} = V \Lambda^{1/2}$ and $\Lambda^{1/2}$ is a diagonal matrix with entries the square roots of the eigenvalues of $\Sigma$. It follows that when $X$ is roughly Gaussian, with mean and covariance matrix $m$ and $\Sigma$, letting $\Sigma^{-1/2} = (\Sigma^{1/2})^{-1}$, transformed variable $X' = \Sigma^{-1/2}(X - m)$ should be close to an isotropic Gaussian. By the change of variables formula:
\begin{equation*}
    \pi_{X'}(x') = \pi_X(x) \left| \frac{\dif x}{\dif x'} \right| = \pi_X\big( \Sigma^{1/2}x' + m \big) |\Sigma^{1/2}|.
\end{equation*}
In computing the gradient and Laplacian of the energy of the transformed distribution, one must use the chain-rule to take into account the matrix $\Sigma^{1/2}$. Samples obtained from $\pi_{X'}$ may be transformed to have distribution $\pi_X$.

In most of the examples presented in this paper, the target distribution undergoes a pre-transformation as above, with $m$ and $\Sigma$ estimated by a Laplace approximation. For notational simplicity, we will continue to refer to sampling random variable $X$ with distribution $\pi$, even when in actual fact we are sampling the transformed distribution $\pi_{X'}$ corresponding to transformed variable $X'$. We make the Laplace approximation using the ``optim'' function in R \citep{R2021}, which uses numerical methods to find the mode of $\pi$ and the Hessian matrix of $\log \pi$ at the mode.

%%%%%%%%%%%%%%%%%%%%%%%%%%%
% Logistic Regression Model
%%%%%%%%%%%%%%%%%%%%%%%%%%%

\section{Logistic Regression Model of Breast Cancer} \label{appendix:log_reg}

The data \citep{mangasarian1990cancer} was obtained from the University of Wisconsin Hospitals, Madison. The response is whether the breast mass is benign or malignant. Predictors are features of an image of the breast mass. The model has dimension $d=10$.
We used a Gaussian product prior with variance $\sigma^2=400$. Following \cite{gelman2008weakly}, we scaled the data so that response variables were defined on $\{-1, 1\}$, non-binary predictors had mean 0 and standard deviation 0.5, while binary predictors had mean 0 and range 1. The posterior distribution was transformed based on its Laplace approximation, as described by Appendix~\ref{sec:targ_pre-transform}.

As described in Section 5.1, we simulated 100 samples paths of an Adaptive Restore process, with stationary distribution the posterior distribution of the logistic regression model of breast cancer. Each sample path was generated for a simulation time of $T=10^6$, then the samples divided into 10 batches. Figure \ref{fig:lr_boxplots} shows boxplots of batch estimates of $\E[X_i]$ and $\E[X_i^2]$ for $i=1,\dots,10$.

\begin{figure}
    \centering
    \begin{subfigure}{.8\textwidth}
        \centering
        \includegraphics[width=\linewidth]{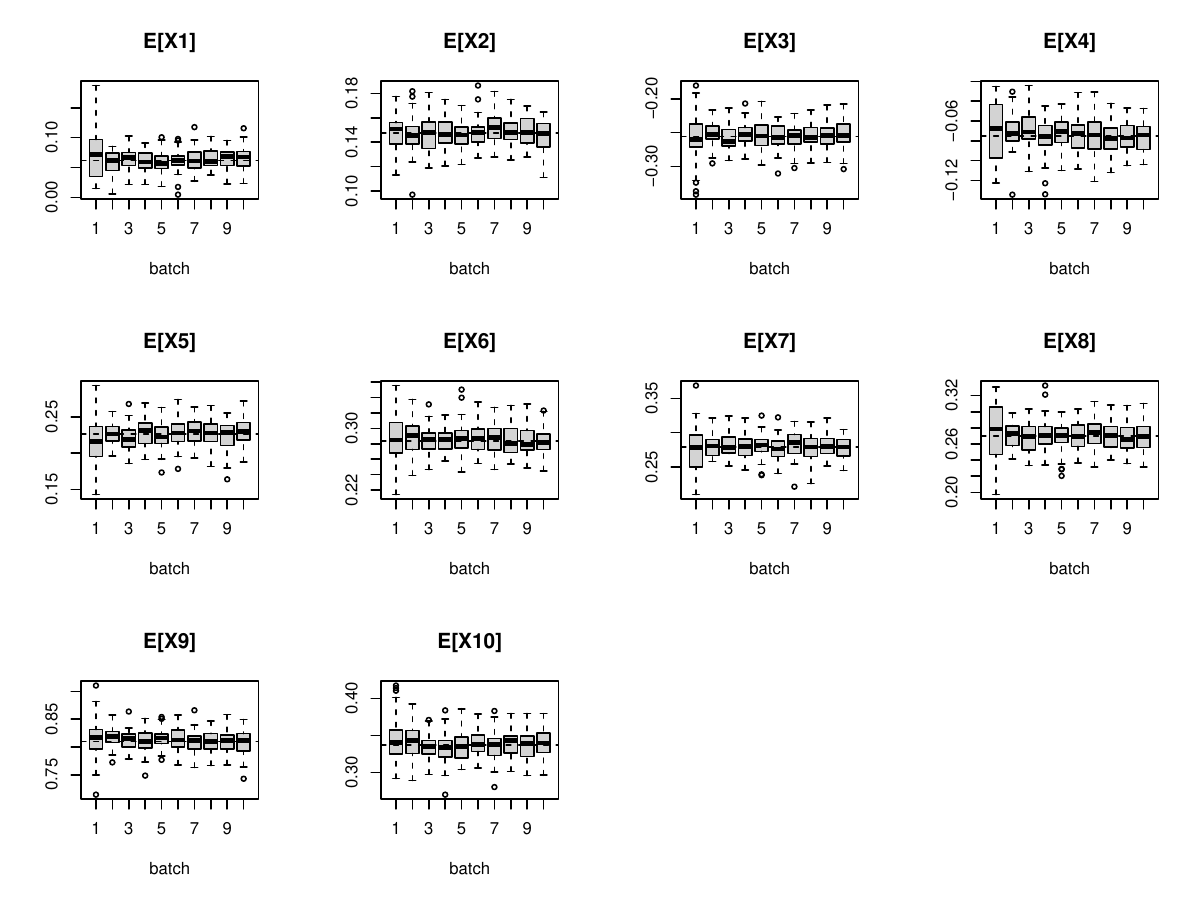}
        \caption{Boxplots of batch estimates of $\E[X_i]; i=1,\dots,10$.}
    \end{subfigure}
    \begin{subfigure}{.8\textwidth}
        \centering
        \includegraphics[width=\linewidth]{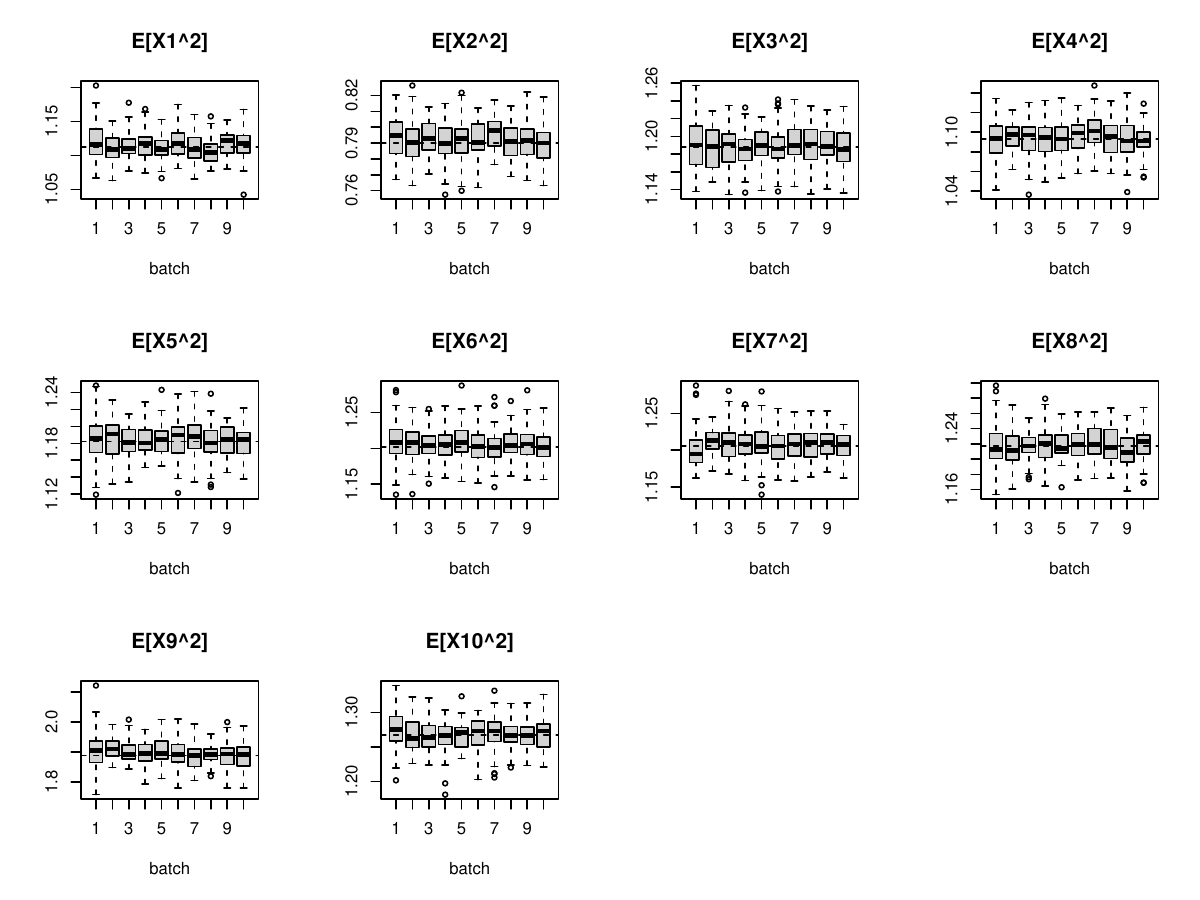}
        \caption{Boxplots of batch estimates of $\E[X_i^2]; i=1,\dots,10$.}
    \end{subfigure}
    \caption{Boxplots of batch estimates of $\E[X_i]$ and $\E[X_i^2]$ for $i=1,\dots,10$ of the posterior of a Logistic Regression model of breast cancer, computed using 100 samples paths of Adaptive Restore processes.}
    \label{fig:lr_boxplots}
\end{figure}

Figure \ref{fig:lr_mse} displays the MSE of estimates of $\E[X_i]$ and $\E[X_i^2]$, for $i=1,\dots,10$, for four types of Restore processes (as described in section 5.1). The burn-in time was $b = 3 \times 10^5$ and the total simulation time $T = 4 \times 10^5$.

\begin{figure}
    \centering
    \begin{subfigure}{.8\textwidth}
        \centering
        \includegraphics[width=\linewidth, trim={0 0.5cm 0 1cm}, clip]{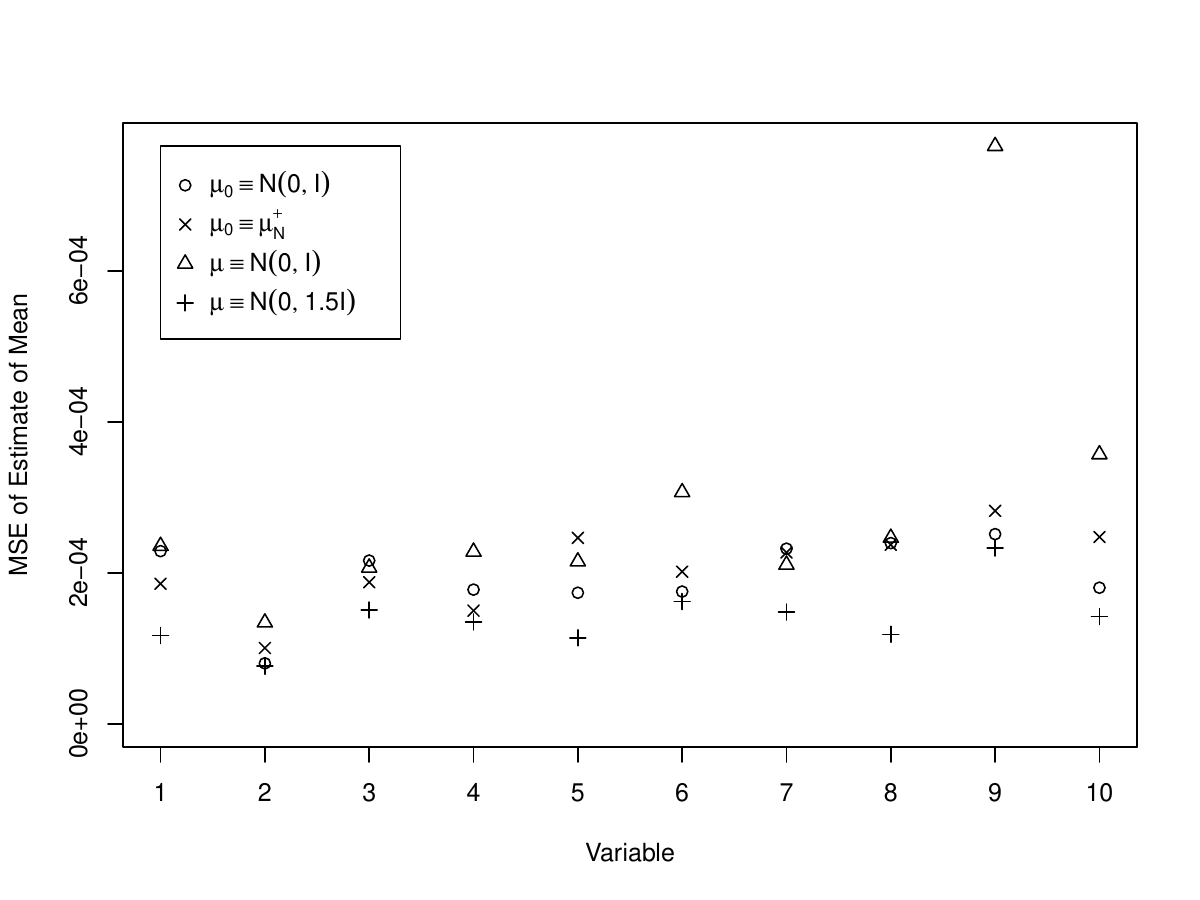}
        \caption{MSE of estimates of $\E[X_i]; i=1,\dots,10$.}
    \end{subfigure}
    \begin{subfigure}{.8\textwidth}
        \centering
        \includegraphics[width=\linewidth, trim={0 0.5cm 0 1cm}, clip]{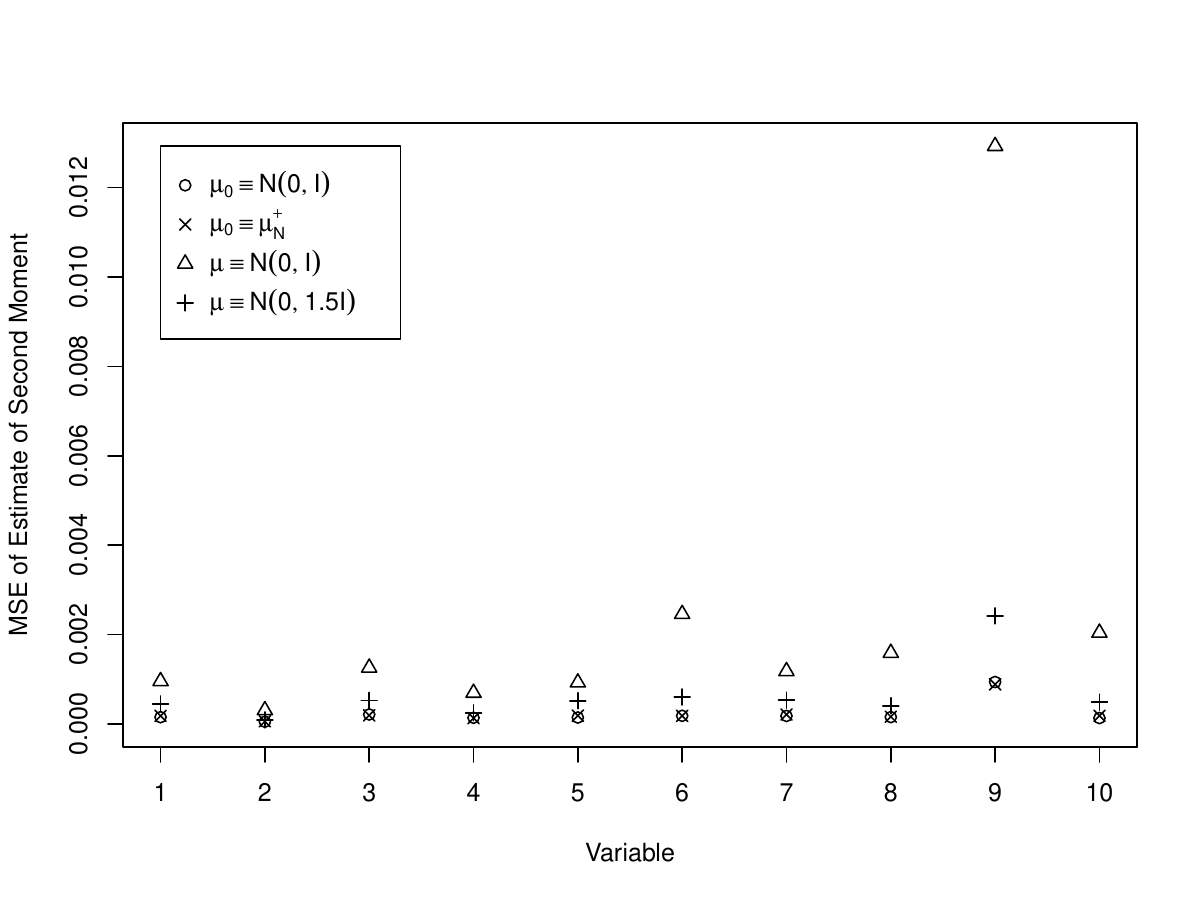}
        \caption{MSE of estimates of $\E[X_i^2]; i=1,\dots,10$.}
    \end{subfigure}
    \caption{MSE of estimates of $\E[X_i]$ and $\E[X_i^2]; i=1,\dots,10$; for different types of Restore processes. Circles and crosses ($\times$) show MSEs for Adaptive Restore processes, with $\mu_0 \equiv \N(0,I)$ and $\mu_0 \equiv \muMin_{\N}$ respectively. Triangles and crosses ($+$) show MSEs for Standard Restore processes, with $\mu \equiv \N(0,I)$ and $\mu \equiv \N(0, 1.5 I)$ respectively.}
    \label{fig:lr_mse}
\end{figure}

%%%%%%%%%%%%%%%%%%%%%%%%%%%%%%%%%%%%%%%%%%%%%%%%%%%%%%%%%%%%%%%%%
% Justification for the minimal restore regeneration distribution
%%%%%%%%%%%%%%%%%%%%%%%%%%%%%%%%%%%%%%%%%%%%%%%%%%%%%%%%%%%%%%%%%

\section{Justification for the minimal regeneration distribution} \label{sec:justification_min_regen_dist}

% NOTE: a long section heading shifts the box on the page that contains text

Throughout the paper, we consider targeting the \textit{minimal regeneration distribution} within our Adaptive Restore algorithm. The reason for this is that we wish to minimise the computational cost per unit time of the algorithm. In subsection 2.4.3, we stated that an advantage of using $\pkap$ for Restore simulation is that to ensure $\P[\pkap < \minRegenRateTrunc] < 1 - \epsilon$ is satisfied, for $\epsilon > 0$ a small constant (e.g. $\epsilon = 0.001$), $\minRegenRateTrunc$ scales logarithmically with dimension $d$. To see this, consider $X \sim \mathcal{N}(0, I)$. Then,
\begin{align*}
    \P[\pkap(X) < \minRegenRateTrunc] = \P[0.5(x^Tx - d) < \minRegenRateTrunc] = \P\left[\sum_{i=1}^d x_i^2 < 2 \minRegenRateTrunc + d\right] = \P[Q < 2 \minRegenRateTrunc + d],
\end{align*}
for $Q \sim \chi_d^2$.
Figure \ref{fig:kappa_bar_guidance} shows $\minRegenRateTrunc$ so that $\P[\pkap(X) < \minRegenRateTrunc] < 1-\epsilon$ for $\epsilon=0.01, 0.001, 0.0001$. Thus for a 100-dimensional Gaussian target distribution, the truncation level $\minRegenRateTrunc = 30$ would likely be appropriate. As a caveat, some functions, such as $f(x)=x^T x$, are very sensitive to $\minRegenRateTrunc$ so an even more conservative choice of $\minRegenRateTrunc$ might be necessary. Furthermore,  $\tilde{\kappa}$ may prove impossible to derive in realistic situations.
\begin{figure}
    \centering
    \includegraphics[width=0.6\linewidth, trim={0 1cm 0 2cm}, clip]{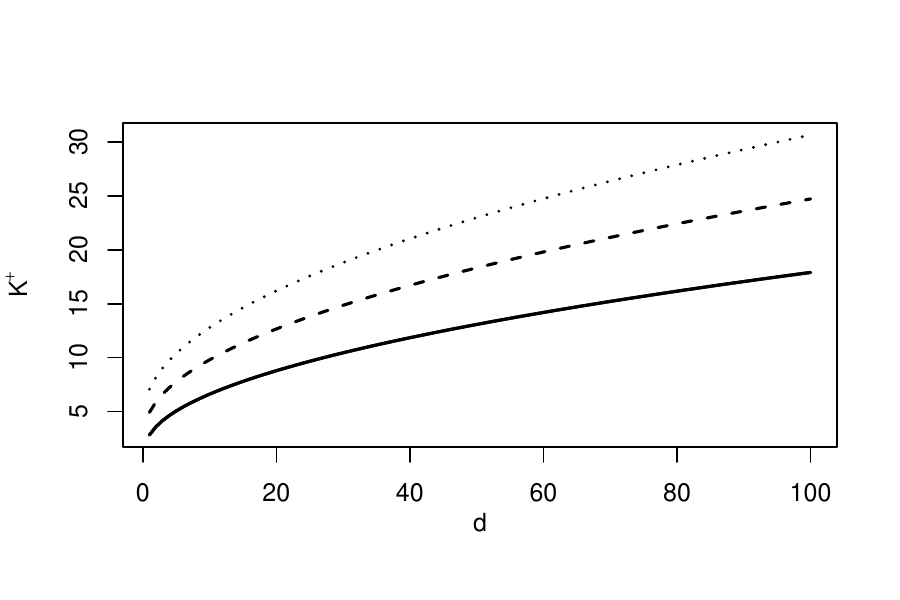}
    \caption{For $\pi \equiv \mathcal{N}(0,I)$, plot of $\minRegenRateTrunc$ so that $P[\pkap(X) < \minRegenRateTrunc] = 1 - \epsilon; \epsilon=0.01, 0.001, 0.0001$.}
    \label{fig:kappa_bar_guidance}
\end{figure}

However, reducing regenerations could conceivably make Markov chains mixing slower. In this section, we consider a stylised example in a formal asymptotic analysis for large dimensions where we demonstrate that the minimal regeneration distribution in fact leads to $O(1)$ convergence time per unit regeneration.

Suppose that we are targeting $\pi= \mathcal N_d ({\bf 0}, I_d)$ and for a random variable $X$ on $\Rd$ set $R = |X |$. We consider minimal regeneration Restore with Brownian motion local dynamics. Notice that by simple calculations we have that
$$
    \frac{L^\dagger \pi(x)}{ \pi(x) } = \frac{|x |^2 -d }{ 2}
$$
and
$$
    {\tilde \kappa}(x) = \left(\frac{ L^\dagger \pi }{ \pi }\right)_+ = \frac{(|x |^2 -d)_+}{ 2}\ .
$$
Furthermore we have that
$$
    C^+{\mu }^+(x ) = (L^\dagger \pi )_- = \frac{{(d-|\bx |^2 )_+} \,\mathrm e^{-\| \bx \|^2/2 } }{ (2\pi )^{d/2}}.
$$
 Now given the spherical symmetry of this example, we marginalise to the radial component:
$$
    C^+ {\tilde \mu }_R(r) = \frac{r^{d-1}{( d-r^2)_+} e^{-r^2/2 } }{ (2\pi )^{d/2}}.
$$
Now to understand how this minimal regeneration Restore behaves  as $d \to \infty $, we must rescale $R$. Note that for large $d$, $R \approx N(d^{1/2}, 1/2)$. Therefore set $Y:=R-d^{1/2}$ and translate the 
regeneration distribution (and rate) to the $Y$ space.
\begin{align*}
    C^+ {\tilde \mu }_Y(y) &= \frac{(y+d^{1/2})^{d-1}{(d-(y+d^{1/2})^2 )_+}\, \mathrm e^{-(y+d^{1/2})^2/2 } }{ (2\pi )^{d/2}}, \ \ \ \hbox{ for } y>0, \\
    &\propto \left(1 + \frac{y }{ d^{1/2}} \right)^{d-1} {( -2y-y^2/d^{1/2} )_+}\,\mathrm e^{-y^2/2} \mathrm e^{-y d^{1/2}}.
\end{align*}
Now the local dynamics of $R$ are well-known to follow the Bessel$(d)$ process, so that for large $d$ we have that the local dynamics for $Y$ are
\begin{align*}
    \dif Y_t &= \dif B_t + \frac{d - 1}{2d^{1/2}} \dif t, \\
    &\approx  \dif B_t + \frac{d^{1/2} }{2} \dif t \ .
\end{align*}
Moreover the regeneration rate is
\begin{align*}
    {\tilde \kappa} &= \frac{1}{2} (\|x \|^2 -d)_+ =  \frac{1}{2} (y^2 + 2yd^{1/2})_+
    \approx y_+ d^{1/2}\ .
\end{align*}
Since the regeneration rate and local dynamics are both $O(d^{1/2})$, so further slowing down the process by a factor of $d^{1/2}$, we finally set $Z_t = Y_{d^{-1/2}t}$, and ignoring terms negligible for large $d$, the local dynamics reduces to
$$
    \dif Z_t \approx \frac{1}{2} \dif t,
$$
namely a constant velocity deterministic flow; the regeneration rate reduces to $z_+$, while the minimal regeneration density is 
$$
    {\tilde \mu }_Z(z) \propto (-z)_+e^{-z^2}\ .
$$
For a formal proof of this final statement, see \cite[Theorem~5.6.1]{Wang2020thesis}.

\begin{remark}
    Note that this non-reversible Markov process can be readily checked to permit $\mathcal{N}(0, 1/2)$ as stationary distribution as expected, by showing that its generator applied to an arbitrary $L^2$ function has mean zero with respect to $\mathcal{N}(0, 1/2)$.
\end{remark}
Convergence of $Z$ is $O(1)$ in $d$-dimensions, as is its regeneration rate. Therefore the minimal regeneration distribution does indeed achieve $O(1)$ convergence per regeneration in high-dimensional contexts.

In contrast, consider a strategy of attempting to construct a regeneration distribution which is close to $\pi $.
If this can be achieved exactly, the resulting Restore algorithm will work extremely well. Of course, however, this will not be achievable in practice. Suppose in the above example we instead manage to achieve a regeneration distribution
$$
    \mu (x) \sim \mathcal N_d({\bf 0}, (1+\epsilon) \mathbf I_d).
$$
Now without repeating the entire analysis, $R$ concentrates like $(1+\epsilon )^{1/2} d^{1/2}$, while the local dynamics propels the process it to even higher values and exponentially large regeneration rates.
Moreover, the probability of regenerating a value of $Z<1$ (where half its stationary mass lies) is exponentially small in $d$ (following from elementary estimates on the $\chi^2_d$ distribution). So regardless of how small $\epsilon $ is, convergence per regeneration is exponentially slow as a function of dimension.

Instead one might hope that using an {\em under}-dispersed regeneration distribution might be more robust. To investigate this we instead set
$$
    \mu (x) \sim \mathcal N_d({\bf 0}, (1-\epsilon) \mathbf I_d).
$$
for some $\epsilon >0$ and again propose to use Brownian motion local dynamics. In this case we see that
$$
\kappa (x) = \|x\|^2 - d + C e^{-\xi \|x \|^2 /2},
$$
for $\xi = \epsilon /(1-\epsilon )$. This $\kappa$ is minimised at 
$$
\|x\|^2 = { 2 \log (C \xi /2 )
\over \xi },
$$
with minimised value $2/\xi - d +2 \log (C\xi /2)/\xi $.
In order to make this non-negative we therefore require
$$
C\ge {2 \exp\{ \xi d /2 -1\}
\over \xi }
\ .
$$
Thus unless $\xi $ (and hence $\epsilon $) is smaller than $2/d$, $C$ will be forced to be exponentially large in $d$ and thus leading to exponential complexity of the algorithm once more.
%Note that this argument assumed $\epsilon >0$ but a similar argument can hold for underestimation of the variance. 

As a result of these calculations, we see that  the strategy of directly targeting $\pi $ as the regeneration distribution lacks robustness in its convergence.

The overall conclusion of these calculations is that targeting the minimal regeneration distribution has the best chance of breaking the worst effects of the curse of dimensionality, thus justifying the strategy we take in this paper.

%%%%%%%%%%%%%%%%%%%%%
% Further Experiments
%%%%%%%%%%%%%%%%%%%%%

\section{Further Experiments}

We include some further experiments, which serve to check the correctness of the implementation: a transformed Beta distribution and a multivariate Gaussian distribution.

%%%%%%%%%%%%%%%%%%%%%%%%%%%%%%%
% Transformed Beta Distribution
%%%%%%%%%%%%%%%%%%%%%%%%%%%%%%%

\subsection{Transformed Beta Distribution}

\label{appendix:transformed_beta}

We experiment with sampling from a distribution with density
\begin{equation*}
    \pi(x) = 6 e^{2x} ( e^x + 1 )^{-4}.
\end{equation*}
This distribution is derived from the transformation of a Beta distribution. Consider $X' \sim \text{Beta}(2,2)$, so that $\pi_{X'}(x') \propto x'(1-x')$ for $\ x' \in [0,1]$. Let $X$ be defined by the logit transformation of $X'$, that is $X = \log\left(\frac{X'}{1 - X'}\right)$, so that $X$ has support on the real line. The inverse of this transformation is $X' = \frac{e^X}{e^X + 1}$, so the Jacobian is $\frac{dx'}{dx} = \frac{e^x}{(e^x + 1)^2}$. Thus $X$ has density:
\begin{equation*}
    \pi(x) = 6 \frac{e^x}{e^x + 1} \left(1 - \frac{e^x}{e^x + 1} \right) \frac{e^x}{(e^x + 1)^2} = \frac{6 e^{2x}}{(e^x + 1)^4}.
\end{equation*}
We make no further transformation of the target. The partial regeneration rate is
\begin{equation*}
    \tilde{\kappa}(x) = \frac{4e^{2x} - 12e^x + 4}{2(e^x + 1)^2}.
\end{equation*}

It happens that $\tilde{\kappa}(x) < 2, \forall x \in \mathbb{R}$, which makes this distribution a useful test case, since an Adaptive Restore process may be efficiently simulated without any truncation of the regeneration rate. In addition, the first and second moments, 0 and $(\pi^2 - 6)/3$, may be computed analytically. Here, $\nK = 0.5$. Taking $\mu_0 \equiv \mathcal{N}(0, 1)$ we simulated 100 Adaptive Restore processes with Short-Term Memory, with parameters $T=10^5, \outputRate=2, a=10, \nCloud=10^4, \nForget=2$. For each sample path, the first half of the output states were burnt. 54 of the 100 estimates of the second moment were greater than the exact second moment; this indicates the processes have (approximately) converged to the correct invariant distribution.

%%%%%%%%%%%%%%%%%%%%%%%%%%%%%%%%%%%%
% Multivariate Gaussian distribution
%%%%%%%%%%%%%%%%%%%%%%%%%%%%%%%%%%%%

\subsection{Multivariate Gaussian distribution}

We use an Adaptive Restore process with Short-Term Memory to sample a multivariate Gaussian distribution with dimension $d=10$, mean $0.5 \times \mathbbm{1}_{10}$, variances $(0.92, 0.94, \dots, 1.10)$ and pairwise covariances 0.5. The parameters used were $\nK = -5.05$ (3.s.f), $\minRegenRateTrunc=11.2$ (3.s.f), $\outputRate=1, a=10, \nCloud=10^4, \nForget=2$ and $T=2 \times 10^5$. The first half of the samples was removed and the remaining samples used to estimate $\E[X^T X]$. The MSE of the estimates was $5.32 \times 10^{-4}$ (3.s.f). Of the 100 estimates, 49 exceeded the true value of $\E[X^T X]$.

%%%%%%%%%%%%%%%%%%%%%%%%%%%%%%%%%%%%
% Hierarchical Model of Pump Failure
%%%%%%%%%%%%%%%%%%%%%%%%%%%%%%%%%%%%

\section{Hierarchical Model of Pump Failure}

Figure \ref{fig:boxplots_pump_mu_gauss} displays boxplots of batch estimates of $\E[X_i]$ and $\E[X_i^2]; i=1,\dots,11$; for $\pi$ the posterior of a hierarchical model of pump failure. 100 samples paths of Adaptive Restore processes were simulated (as described in Section 5.2) and the samples from each path divided into 10 batches.

\begin{figure}
    \centering
    \begin{subfigure}{\linewidth}
        \centering
        \includegraphics[width=.75\linewidth]{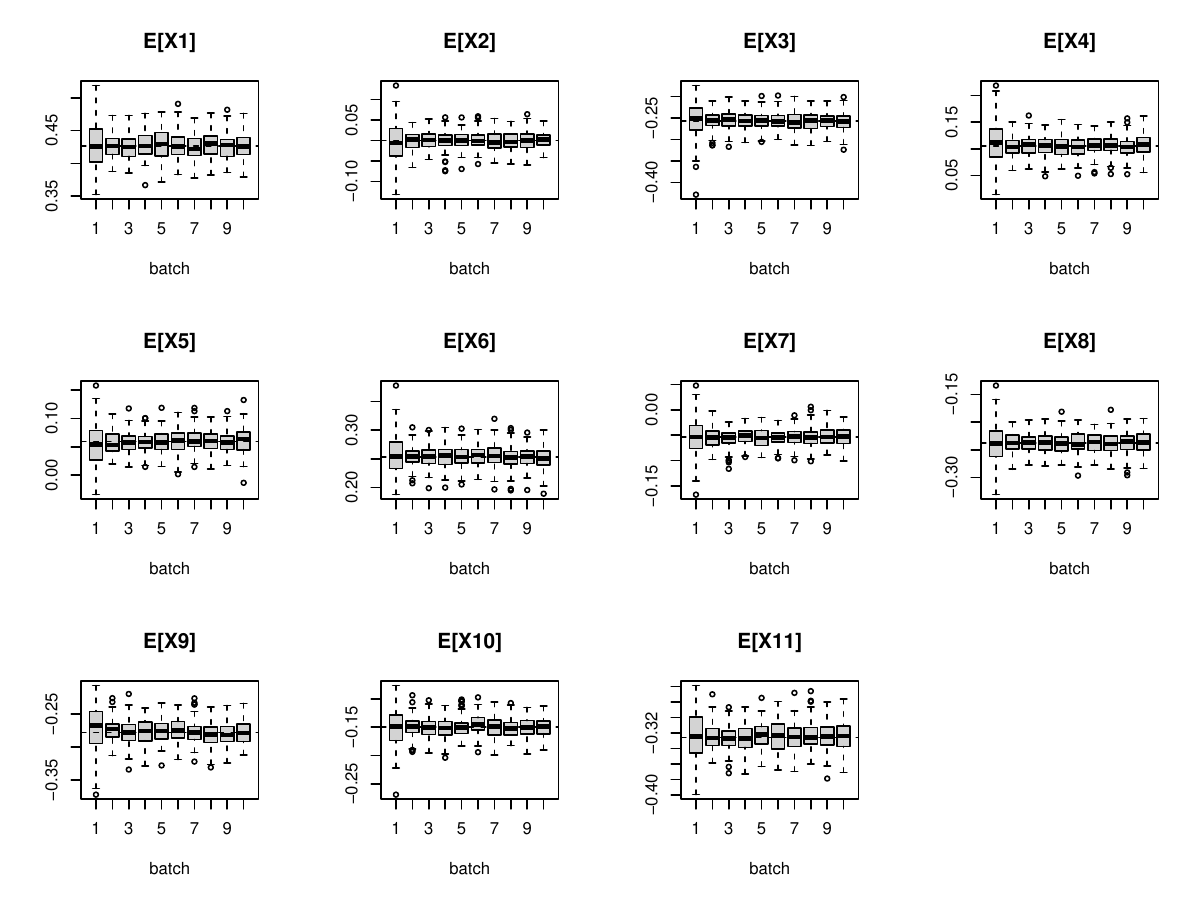}
        \caption{Estimates of $\E[X_i]$ for $i=1,\dots,11$.}
        \label{fig:mo1_boxplots_pump_mu_gauss}
    \end{subfigure}\\
    \begin{subfigure}{\linewidth}
        \centering
        \includegraphics[width=.75\linewidth]{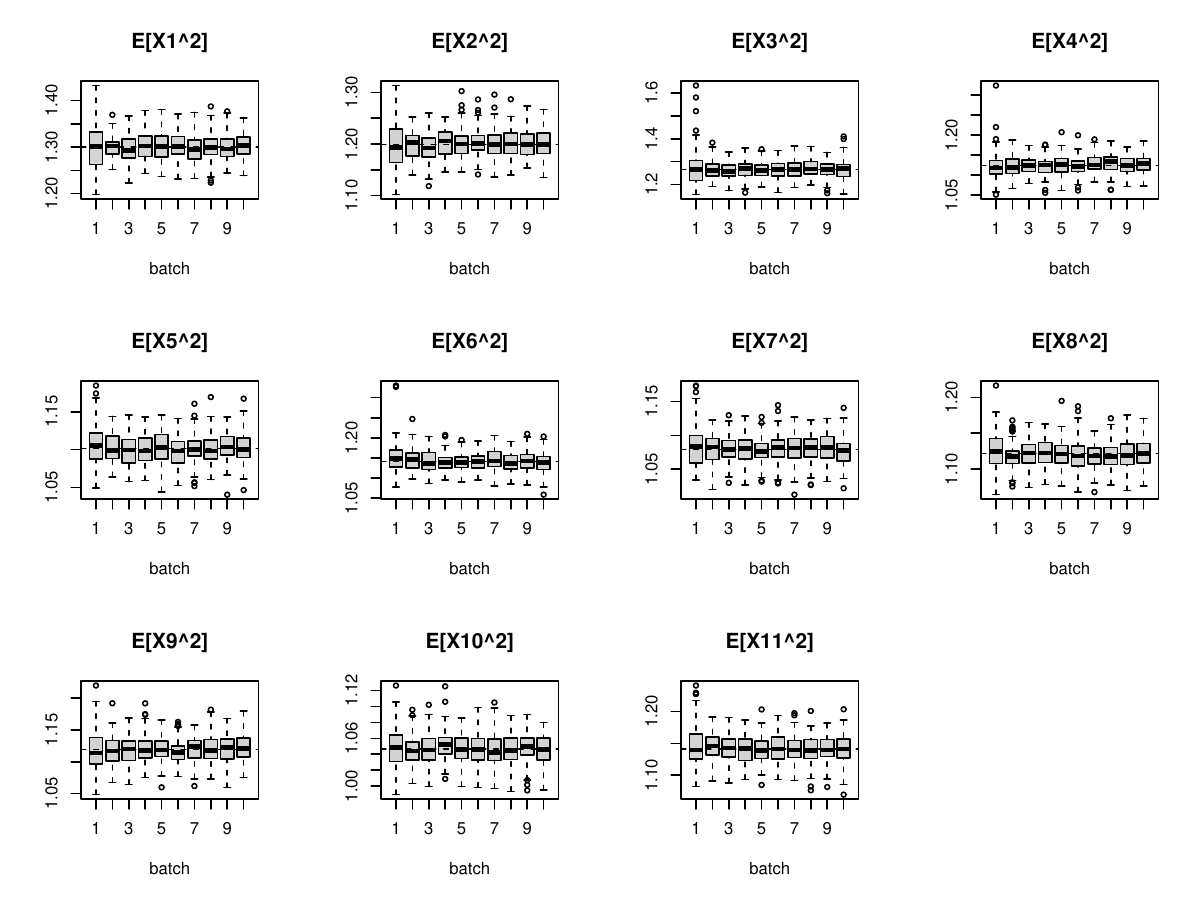}
        \caption{Estimates of $\E[X_i^2]$ for $i=1,\dots,11$.}
        \label{fig:mo2_boxplots_pump_mu_gauss}
    \end{subfigure}
    \caption{Boxplots of batch estimates of $\E[X_i]$ and $\E[X_i^2]; i=1,\dots,11$; for the posterior of the hierarchical model of pump failure; generated using 100 Adaptive Restore processes. A horizontal dashed line shows a very accurate approximation of the true value, computed using a long Markov chain.}
    \label{fig:boxplots_pump_mu_gauss}
\end{figure}

For the same posterior, Figure \ref{fig:pump_rstr_compare_mse} displays MSEs of $\E[X_i]$ and $\E[X_i^2]$; $i=1,\dots,11$; for different types of Restore processes, as described in section 5.2.

\begin{figure}
    \centering
    \begin{subfigure}{\linewidth}
        \centering
        \includegraphics[width=.6\linewidth, trim={0 0.5cm 0 1.5cm}, clip]{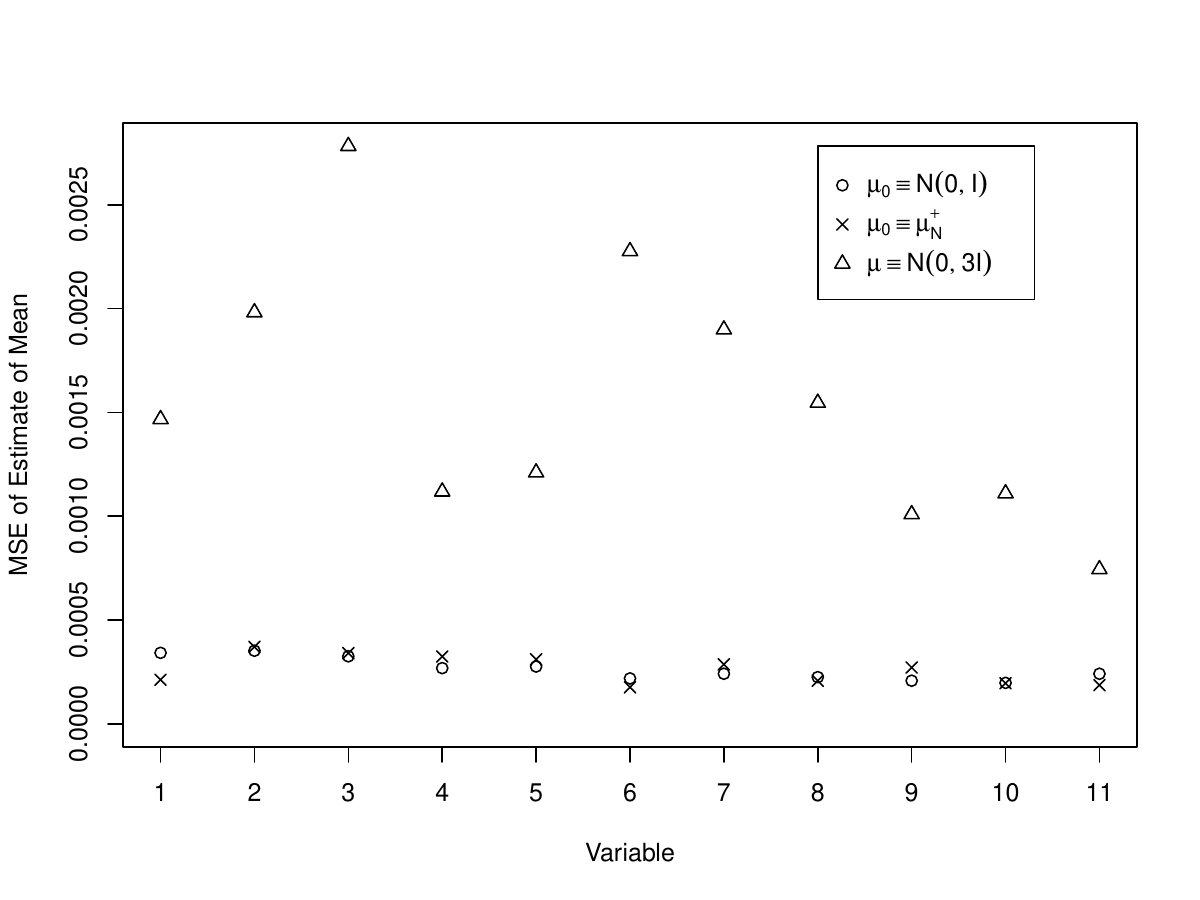}
        \caption{MSEs of estimates of $\E[X_i]; i=1,\dots,11$.}
        \label{fig:pump_rstr_compare_mse_mo1}
    \end{subfigure}\\
    \begin{subfigure}{\linewidth}
        \centering
        \includegraphics[width=.6\linewidth, trim={0 0.5cm 0 1.5cm}, clip]{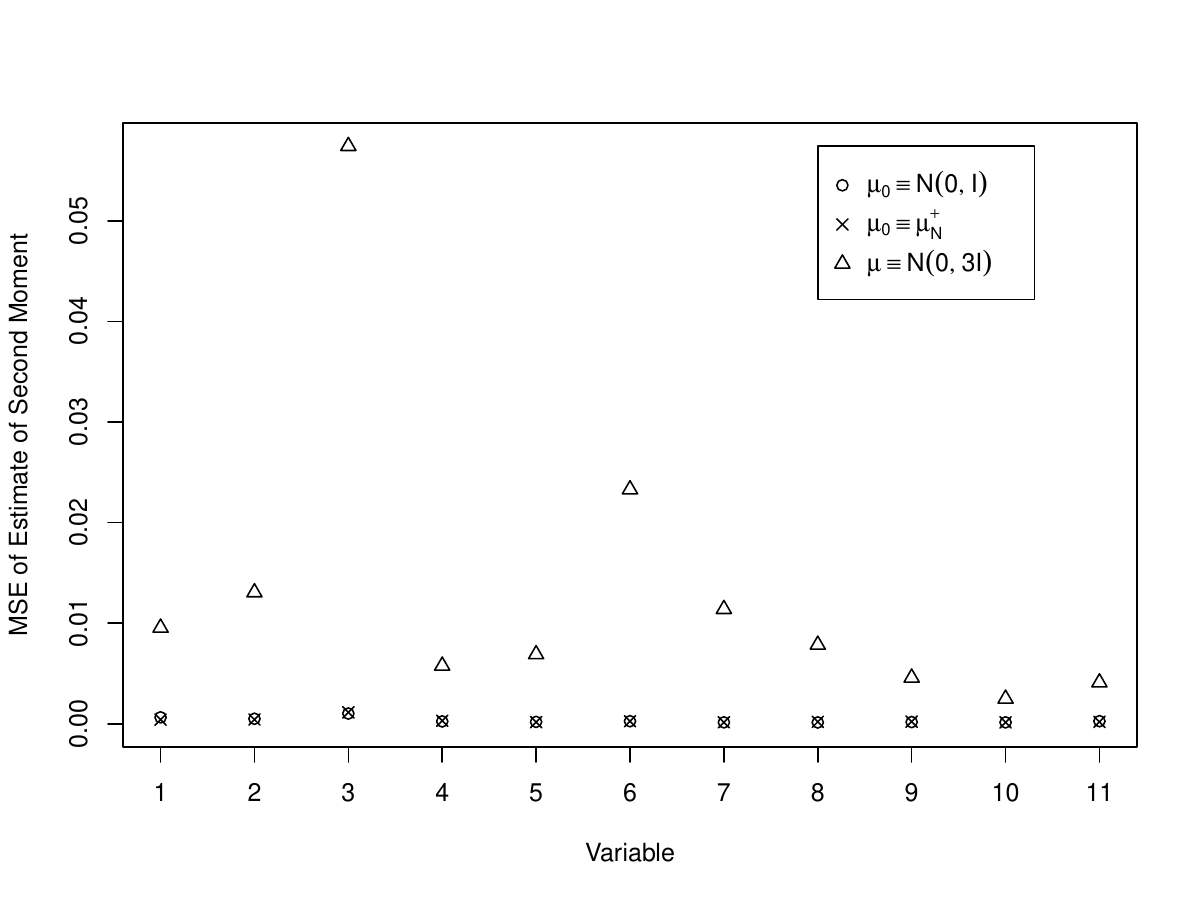}
        \caption{MSEs of estimates of $\E[X_i^2]; i=1,\dots,11$.}
        \label{fig:pump_rstr_compare_mse_mo2}
    \end{subfigure}
    \caption{MSEs of estimates of $\E[X_i]$ and $\E[X_i^2]$ for $i=1,\dots,11$ for the posterior distribution of the hierarchical model of pump failure. The MSEs were computed by simulating 100 samples paths. Circles and crosses correspond to Adaptive Restore processes with $\mu_0 \equiv \N(0,I)$ and $\mu_0 \equiv \muMin_{\N}$ respectively. Triangles correspond to Standard Restore processes with $\mu \equiv \N(0, 3I)$.}
    \label{fig:pump_rstr_compare_mse}
\end{figure}

\end{appendix}

%\bibliographystyle{imsart-nameyear.bst}
%\bibliography{adaptive_restore.bib}

\end{document}

%% file: macros.tex
\newcommand{\R}{\mathbb R}
\newcommand{\Rd}{\mathbb R^d}
\newcommand{\E}{\mathbb E}
\renewcommand{\P}{\mathbb P}

\newcommand{\A}{\mathcal A}
\newcommand{\N}{\mathcal{N}}

\newcommand{\cX}{\mathcal{X}}
\newcommand{\TV}{\mathrm{TV}}

\newcommand{\nkap}{\kappa^-}    % Rate at which samples from the min regen dist arrive
\newcommand{\nK}{K^-}           % Supremum of the rate of the PP at which samples from
\newcommand{\pkap}{\kappa^+}    % Minimal regeneration rate
\newcommand{\pK}{K^+}           % Supremum of the minimal regeneration rate
\newcommand{\K}{K}              % Supremum of the regeneration rate
\newcommand{\muMin}{\mu^+}      % Minimal regeneration distribution
\newcommand{\CMinimum}{C^+}     % Minimal regeneration constant
\newcommand{\muMinEvent}{\zeta} % Time to the next next event from the PP at which samples
\newcommand{\muMinDomEvent}{\tilde{\zeta}} % Time to the next event from the PP dominating
\newcommand{\regenRateTrunc}{\mathcal{K}}      % Truncation of the Regeneration Rate
\newcommand{\minRegenRateTrunc}{\mathcal{K}^+} % Truncation of the Minimal Regeneration Rate
\newcommand{\tourLength}{\tau}                 % Tour length
\newcommand{\timeToNextPotentialRegen}{\tilde{\tau}}

\newcommand{\outputRate}{\Lambda_0}

\newcommand{\subKernel}{Q}      % originally K
\newcommand{\probRegenBeforeMuMin}{q} % originally k
\newcommand{\properKernel}{Q_0} % originally K_0
\newcommand{\augKern}{\bar{Q}}  % originally \bar{K}
\newcommand{\mcKern}{Q_\mu}     % originally K_\mu
\newcommand{\mcKernDscrtMeas}{Q_{\mu_n}} % originally K_{\mu_n}

\newcommand{\nCloud}{n_{\text{cloud}}}
\newcommand{\nForget}{n_{\text{forget}}}
\newcommand\bx{{\bf x}}